\documentclass[pra,aps,twocolumn,superscriptaddress,longbibliography]{revtex4-1}
\usepackage[colorlinks=true, citecolor=red, urlcolor=blue ]{hyperref}
\usepackage{graphicx}% Include figure files
\usepackage{bm}
\usepackage{amsmath,amsfonts}
\usepackage{amsthm}
\usepackage{hyperref}
\usepackage{xcolor}
\hypersetup{
    urlcolor=magenta,
    citecolor=blue
           }

\usepackage{braket}
\theoremstyle{plain}
\usepackage{color}
\usepackage{amssymb}
\usepackage{amsthm}
\usepackage{amsfonts}
\usepackage{float}
\usepackage{tabularx}
\usepackage{graphicx}

\usepackage{mathtools}
\usepackage{esvect}
\usepackage{wrapfig}
\usepackage{amsthm}
\usepackage{verbatim}
\usepackage{bbm}
\usepackage[normalem]{ulem}

\usepackage{enumitem}
\usepackage{fmtcount}
\usepackage{booktabs}
\usepackage{csquotes}
\usepackage{epsfig}

\usepackage{tabularx}
\usepackage{graphicx}
\usepackage{amsmath}
\usepackage{braket}
\usepackage{latexsym}
\usepackage{bm}
\usepackage{graphics,epstopdf}
\usepackage{enumitem}
\usepackage{fmtcount}
\usepackage{booktabs}
\usepackage{csquotes}
\usepackage{epsfig}

\theoremstyle{plain}

\def\bea{\begin{eqnarray}}
\def\eea{\end{eqnarray}}
\def\ba{\begin{array}}
\def\ea{\end{array}}

\def\beq{\begin{equation}}
\def\eeq{\end{equation}}

\usepackage[normalem]{ulem}
\usepackage{float}
\usepackage{graphicx}  
\usepackage{dcolumn}          
\usepackage{amssymb}
\usepackage{appendix}
\usepackage{physics}   
\usepackage{mathtools}
\usepackage{esvect}
\usepackage{wrapfig}
\usepackage{amsthm}
\usepackage{verbatim}
\usepackage{bbm}

\usepackage[mathscr]{euscript}
\def\Tr{\operatorname{Tr}}

\def\({\left(}
\def\){\right)}
\def\[{\left[}
\def\]{\right]}

\newcommand{\mc}[1]{\mathcal{#1}}

\newtheorem{theorem}{Theorem}

\newtheorem{criterion}{Criterion}
\newtheorem{definition}{Definition}

\newtheorem{lemma}{Lemma}

\begin{document}

%\title{New approach of nonlocal advantage of quantum coherence and its monogamy}

\title{Exclusion principle for nonlocal advantage of quantum coherence}

\author{Priya Ghosh}
\affiliation{Harish-Chandra Research Institute,  A CI of Homi Bhabha National Institute, Chhatnag Road, Jhunsi, Prayagraj  211019, India}
\author{Mahasweta Pandit}
\affiliation{Institute of Theoretical Physics and Astrophysics, Faculty of Mathematics, Physics and Informatics, University of Gda\'nsk, 80-308 Gda\'nsk, Poland}
\affiliation{Departamento de Física, Universidad de Murcia, Murcia E-30071, Spain}
\author{Chirag Srivastava}
\affiliation{Harish-Chandra Research Institute,  A CI of Homi Bhabha National Institute, Chhatnag Road, Jhunsi, Prayagraj  211019, India}
\affiliation{Laboratoire d'Information Quantique, CP 225, Université libre de Bruxelles (ULB),
Av. F. D. Roosevelt 50, 1050 Bruxelles, Belgium}
\author{Ujjwal Sen}
\affiliation{Harish-Chandra Research Institute,  A CI of Homi Bhabha National Institute, Chhatnag Road, Jhunsi, Prayagraj  211019, India}

\begin{abstract}
Coherences in mutually unbiased bases of states of an isolated quantum system follow a complementarity relation. The nonlocal advantage of quantum coherence (NAQC), defined in a bipartite scenario, is a situation in which the average quantum coherences of the ensembles of one subsystem, effected by a measurement performed on the other subsystem, violates the complementarity relation. We analyze two criteria to detect NAQC for bipartite quantum states. We construct a more generalized version of the criterion to detect NAQC that is better than the standard criterion as it can capture more states exhibiting NAQC.  
%and show that m version of the NAQC criterion captures a larger set of states that demonstrate NAQC. 
We prove the local unitary invariance of these NAQC criteria.%, which was absent in the earlier standard NAQC criterion.
%Additionally, we also present a relation to efficiently identify the bipartite states that display NAQC. 
~Further on, we focus on investigating the monogamy properties of NAQC in the tripartite scenario. We check for monogamy of NAQC from two perspectives, differentiated by whether or not the nodal observer in the monogamy relation performs the measurement for the nonlocal advantage. We find in particular that in the case where the nodal observer does not perform the measurement, a strong monogamy relation - an exclusion principle - is exhibited by NAQC.
%We show that in the first case, all three-qubit pure states that exhibit NAQC are monogamous, while in the second case, it may not always be true. Monogamy is guaranteed in the second case if the states are fully separable and bi-separable such that the party corresponding to the coherence measurement shares no entanglement with the other parties.
% In the latter case, we find that fully separable three-qubit pure states and bi-separable states without any entanglement between the A and BC partition will always satisfy NAQC monogamy. 

\end{abstract}

\maketitle

\section{Introduction}
\label{sec1}
%In recent years, fundamental studies of quantum effects have resulted in the identification of quantum coherence (QC)~\cite{review-coherence-1} as a key notion of non-classicality.
In recent years, quantum coherence (QC)~\cite{review-coherence-1} has emerged as a key notion of non-classicality. The underlying concept of QC is the wave-like nature of systems which allows two distinct pure states of a system to interfere coherently with each other, forming a quantum superposition. QC is crucial in areas such as quantum information theory~\cite{review-coherence-1}, metrology~\cite{Giovannetti_2011,Giorda18}, quantum biology~\cite{Plenio_2008,Rebentrost_2009,Lloyd_2011,Huelga_2013,alma9916559633502466}, and quantum thermodynamics~\cite{Rosario13,Lostaglio_2015,PhysRevX.5.021001,PhysRevE.92.042126}. In the literature, there exist different measures of QC, and the resource theory of QC has also been developed~\cite{Aberg06,l1-norm,Winter16,plenio-coherence,Bischof19,Srivastava21,sreetama-coherence,Banerjee21}. However, understanding about the ability to manipulate and utilize QC as a resource is far from complete, especially in the multipartite setting.\\
There exists a complementarity relation for coherences of an isolated qubit system in that the sum of quantum coherences of its states in mutually unbiased bases (MUBs)~\cite{Wootters89,vwani-somshubhro-mub,planat06,MUB-1} is non-trivially bounded from above~\cite{hall,pati}. 
The notion of nonlocal advantage of quantum coherence (NAQC) was introduced in the bipartite scenario to detect steerability, as captured by coherence~\cite{pati}. That is, the average coherence of the ensemble of states of one of the parties, affected by the measurement on the other party, was shown to violate the complementarity bound for quantum coherence of an isolated system.
%States that violate any of the NAQC inequalities possess a set of correlations that can be generated on one of the subsystems by performing local measurements on the other and cannot be achieved by an isolated system.
While a state that shows NAQC is steerable, the converse is not always true, making the set of states exhibiting NAQC strictly smaller than the set of steerable states~\cite{debasis}.

Unlike its classical counterpart, non-classical correlations of multipartite systems have restricted sharability and are referred to as the monogamous property~\cite{ent-monogamy}, for the quantum state and the non-classical shared physical quantity considered.  
%It is an well established topic in the literature and 
%has been studied 
%for a plethora of quantities, e.g., 
For research on monogamy of entanglement, see Refs.~\cite{ent-monogamy,mono-winter,mono-adesso,mono-osborne,mono-chen,mono-lee}, and for that of quantum discord, see~\cite{discord-monogamy,mono-discord-2}. See \cite{ent-review-1} and \cite{kavan-discord-review,discord-review-1} for reviews of entanglement and quantum discord, and ~\cite{sanders-ent-review,discord-monogamy-review-ujjwal} for their monogamy properties.
%Monogamy of any correlation in multipartite system provides the information about the distribution of the correlation among the subsystems of the corresponding system.   Monogamy of many different quantum correlations, viz. entanglement~\cite{ent-monogamy,mono-winter,mono-adesso,mono-osborne,mono-seevinck,mono-chen,mono-lee,mono-dagomir},discord~\cite{discord-monogamy,mono-discord-2}, etc. have already been studied very well.

The current work focuses on investigating the monogamy properties of NAQC, as an integral part of characterizing it as a non-classical correlation. We consider two criteria of NAQC and refer to them as the standard and generalized NAQC criteria. The standard criterion includes measurements only in bases belonging to an arbitrary MUB on one of the subsystems of the bipartite system, while the generalized version includes measurements belonging to an arbitrary set of bases on a subsystem. Both criteria  involve optimizations over the relevant sets of projective measurements (MUBs for standard criterion and arbitrary set of bases for generalized criterion) and the set of MUBs for quantum coherence measurement, which makes both the NAQC functionals local unitary invariant. We show that the generalized NAQC criterion allows to capture more states that exhibit NAQC. We also provide a lower bound of both the NAQC functionals for any bipartite system.

%We then proceed to 
Next, we check the monogamy properties of NAQC which is motivated by the following consideration. The phenomenon of NAQC and its quantifications utilize a single-site quantum characteristic, viz. quantum coherence, to create a two-body physical quantity, which exhibits certain ``nonlocal" aspects. Typical quantum correlations are often connected with ``nonlocality", and almost universally observed to exhibit certain monogamy properties~\cite{sanders-ent-review,discord-monogamy-review-ujjwal}.
%It is therefore interesting to check whether other quantum correlation aspects exist in the two-body quantity. 
With this motivation, we consider the monogamy properties of NAQC, for the set of pure three-qubit states in two separate scenarios.
The first case corresponds to the situation where  the  ensemble-generating measurements are performed on different subsystems of the tripartite system, while we switch parties to define the monogamy relation. The measurements are therefore performed on the ``non-nodal" parties in this case. We show that the generalized NAQC is strongly monogamous in this scenario. Indeed, there appears an ``exclusion principle'' in this case.
%We also present a non-trivial upper bound to the monogamy relation which is twice the maximum value of generalized NAQC obtained by maximizing over all pure two-qubit states. 
%We further generalize the bound in multipartite scenario.
In the second scenario, the nodal party performs the ensemble-generating measurements.
The strong monogamy is no more valid here.
%In this case, we show that all fully separable, and bi-separable states which have no entanglement present between the nodal party and the rest, are monogamous. This guarantees that when monogamy is violated, the states must be the other bi-separable states or states that lie in the Greenberger–Horne–Zeilinger~\cite{ghz-class-1,ghz-class-2} and W-classes~\cite{w-class-1,three-pure-2}.

%the monogamy of the generalized NAQC may or may not be satisfied by the corresponding system
%Quantum state of a physical system can exist in a superposition of a fixed set of states forming a basis in the relevant Hilbert space. Superposition of quantum states is an intriguing feature of quantum mechanics which differentiates it from the classical physics. QC~\cite{review-coherence-1} is used to quantify the amount of superposition in a quantum state with respect to a fixed basis.
%Quantum coherence~\cite{review-coherence-1} is an intriguing quantum resource which can differentiate quantum physics from the classical one. It originates the amount of superposition of states. 
%Quantum incoherent states are the free states of coherence resource theory where incoherent operations are free operations. 
\section{Preliminaries}
\label{sec2}
We begin by introducing some of the key elements essential for analyzing the monogamy of NAQC. The central notion of this work is \emph{quantum coherence} (QC). QC, assuredly a basis-dependent notion, is the 
%\textcolor{red}{
underlying concept of quantum entanglement and other quantum phenomena. It arises from coherent superposition of states of a quantum system. 
%in combination with the quantization of observables. 
%QC has enabled new avenues in quantum technologies, such as laser and its applications.
%and has also been established as a very useful resource. 
Several measures can be opted for measuring QC, such as relative entropy of quantum coherence~\cite{l1-norm}, $l_1$-norm of quantum coherence~\cite{l1-norm}, geometric quantum  coherence~\cite{coherence-streltsov}, etc.

For the current work, we will opt for the \emph{$l_1$-norm of quantum coherence}~\cite{l1-norm}, which, for
%$l_1$-Norm of coherence 
any quantum state $\rho$, is defined as the sum of absolute values of all non-diagonal elements of the state corresponding to a chosen basis $M$, i.e., %\label{eq-co-single}
~$C_{M}(\rho) \coloneqq \sum_{i\neq j} |\rho_{ij}|$,
where $\rho_{ij}$ represents the component of the state $\rho$ for the $i$-th row and $j$-th column, in the computational basis. The $l_1$-norm of coherence of $\rho$ in a given basis is zero if all the non-diagonal elements of the state in the corresponding basis are zero, and such states are termed \emph{incoherent states} for that basis, in the resource theory of coherence.

The subsequent key ingredient to define NAQC is the concept of \emph{mutually unbiased bases}~\cite{Wootters89,vwani-somshubhro-mub,planat06,MUB-1}. Let $\{\lbrace \ket{e_i^a} \rbrace^a\}_i$ represent a set of bases, where the subscript $i$ indicates different bases and $a$ indicates different elements of a basis. %and $\lbrace \ket{e_j^b} \rbrace^b$, 
If for a set of orthonormal bases on the Hilbert space $\mc{H}$ with dimension $d$, the elements of any two bases satisfy the relation
$|\bra{e^a_i}e^b_j \rangle|^2 = \frac{1}{d}$,
for $i\neq j$ and all $a, b \in \lbrace 1, 2, ..., d \rbrace$, then this set of bases is said to form a set of MUBs. There can be at most three bases forming a set of MUBs in a two-dimensional complex Hilbert space. The set of eigenbases of $\sigma_x$, $\sigma_y$, and $\sigma_z$ (Pauli) matrices is an example.
%of bases forming an MUB in the two-dimensional Hilbert space.  
We will exclusively be considering MUBs in the qubit Hilbert space.
%, and will always mean a set of three MUBs when we mention ``a set of MUBs".

At this point, we are ready to set forth the \emph{quantum coherence complementarity relation for a  single-qubit system}, presented in \cite{hall,pati}, which states that the
%In general, any single-qubit system can never cross a bound of sum of QC measured in any mutually unbiased bases, i.e.,
sum of quantum coherences for any qubit, over a set of three bases forming a set of mutually unbiased bases, has the following non-trivial upper bound: 
\begin{align}\label{saral}
    \sum_{i=1}^3 C_{M_i} \leq \sqrt{6}.
\end{align}
Here, $C_{M_i}$ denotes the value of quantum coherence of the corresponding state measured in the basis $M_i$, where $\{M_i,i=1,2,3\}$ forms a set of MUBs.  
%And $\xi$ is just a number dependent on coherence measure (for example, it is $\sqrt{6}$ for $l_1$-norm of coherence measure).
This bound is non-trivial as the maximum value of the $l_1$-norm of coherence of a qubit over all  bases is unity.
%This is known as coherence complementary relations.
%\end{definition}
%For $l_1$-norm of coherence measure, the upper bound of sum of coherence of single qubit system measured in arbitrary mutually unbiased basis is $\sqrt{6}$. 
%Advantage in QC due to steeribility has been observed recently in bipartite systems~\cite{pati}.

This sets the stage for us to briefly recapitulate the notion of nonlocal advantage of quantum coherence. Consider two spatially separated parties, Alice and Bob, sharing a two-qubit state $\rho_{AB}$. Alice performs measurements on her part of the shared state, resulting in ensembles of states of Bob's subsystem. As Alice communicates her measurement outcomes to Bob, the average quantum coherence of the ensembles of states at Bob's end may violate the quantum coherence complementarity relation given in \eqref{saral}, %. Such a violation of coherence complementary relation is called as % nonlocal advantage of quantum coherence 
~exhibiting a ``nonlocal advantage" of QC. Formally, NAQC can be defined as follows.
\begin{definition}[Nonlocal advantage of quantum coherence~\cite{pati}]
    Violation of the quantum coherence complementarity relation for average quantum coherences of one of the subsystems of a bipartite system by measurements on the other subsystem is defined as the nonlocal advantage of quantum coherence.
\end{definition}
%The bipartite states allowing such violation are said to be advantageous for NAQC.
%It has been shown that Alice measurement on her side of the bipartite state followed by coherence measurement on conditional states in any mutually unbiased basis by Bob after classical communication with Alice can violate the coherence complementary relation described above. This is known as nonlocal advantage of quantum coherence (NAQC)~\cite{pati}.

 It has been established that only entangled bipartite states can exhibit NAQC~\cite{pati}, and indeed, non-steerable states cannot show NAQC, so the set of states providing NAQC is a subset of the set of steerable states~\cite{debasis}. For states on $\mathbb{C}^2 \otimes \mathbb{C}^2$ with diagonal correlation matrix, NAQC also captures stronger quantum correlation than ``Bell-nonlocality"~\cite{bell-naqc}. 

Lastly, we mention the standard form of the three-qubit pure state, that will be used in this work. In the computational basis, any three-qubit pure state  $\ket{\Psi}_{ABC}$, up to local unitaries, can be written as
\begin{eqnarray}
\label{eq-three-pure}
    \ket{\Psi}_{ABC} &=& \lambda_0 \ket{000} + \lambda_1 e^{i \beta} \ket{100} + \lambda_2 \ket{101} +\lambda_3 \ket{110}\nonumber \\ &~& +  \lambda_4 \ket{111},
\end{eqnarray}
where $\lambda_i \geq 0$, $\sum_{i=0}^{4} \lambda_i^2 = 1$, $\beta \in [0,\pi]$~
\cite{three-pure-1,three-pure-2,three-pure-3}.
%\textcolor{red}{W class}
%\end{definition}
%All the GHZ-class, W-class, bi-separable, and fully-separable pure states can be expressed by above equation e.g. 
Eq. \eqref{eq-three-pure} represents W-class states when $\lambda_4 = 0$ and $\beta = 0$.
%\begin{equation}
 %\ket{\psi}_\textnormal{W} = \lambda_0 \ket{000} + \lambda_1 \ket{100} + \lambda_2 \ket{101} +\lambda_3 \ket{110}.
%\end{equation}

\section{Detection of nonlocal advantage of quantum coherence}
\label{sec3}
In this section, we present and analyze two criteria  to detect NAQC. 
In the standard criterion to detect NAQC, Alice performs measurements only in the bases belonging to a set of MUBs, $\lbrace \lbrace \Lambda^a_i \rbrace^a \rbrace_i$. She then communicates her measurement settings and outcomes to Bob, creating three ensembles of states, $\mc{E}_i\equiv \{p(\rho_{B|\Lambda^a_i}),~\rho_{B|\Lambda^a_i}\}^a$, at Bob's end, where $p(\rho)$ is the probability of getting the state $\rho$, in the relevant measurement (at Alice). Now consider an independent set of MUBs, viz. $\{M_i\}_i$, and Bob measures the quantum coherence of the ensemble $\mc{E}_i$ in the basis $M_i$, so that the average coherence of ensemble $\mc{E}_i$ in the basis $M_i$ is
%in the basis $M_i$, which is given as 
given by $\sum_a p(\rho_{B|\Lambda^a_i}) C_{M_i}(\rho_{B|\Lambda^a_i})$. Notice that the average of coherences is taken over the outcomes (at Alice) for a given measurement setting $i$ (at Bob). Finally, we sum these average quantum coherences over all values of $i$ to obtain $\sum_{i,a}  p(\rho_{B|\Lambda^a_i}) C_{M_i}(\rho_{B|\Lambda^a_i})$. This sum is to be compared with the  L.H.S. of the complementarity relation \eqref{saral}. This sum can be maximized over the choice of the sets of MUBs, $ \lbrace \Lambda_i  \rbrace_i \equiv \lbrace \lbrace \Lambda^a_i \rbrace^a \rbrace_i$ and $\{M_i\}_i$, to obtain
% which is given by,
%NAQC have been looked from three perspectives which are known as one, two, and three measurement settings~\cite{debasis}. In one measurement settings, 
%%Alice perform three measurement locally on her side and communicate to Bob classically. Then, 
%Bob will measure coherence on conditional states sent by Alice through classical communication with respect to each basis of MUB after each local operation performed by Alice on her side. 
%In two measurement settings,  Alice perform three measurement locally on her side, communicate to Bob classically and Bob will perform two coherence measurement on conditional states after each measurement in Alice side in any mutually unbiased basis.
%Analogously, two and three measurement settings have been defined. NAQC criterion for one measurement settings, of two-qubit state $\rho_{AB}$ will be the following form $\mc{N}_1^\rightarrow (\rho_{AB}) \leq \xi$ with
%%nonlocal advantage of quantum coherence for any bipartite state $\rho_{AB}$ in one measurement settings, can be defined as
\begin{align}
    \label{eq-naqc}
    \mathcal{N}^\rightarrow (\rho_{AB}) \coloneqq \max_{\{M_i\}_i,\{\Lambda_i\}_i} \sum_{i,a}  p(\rho_{B|\Lambda^a_i}) C_{M_i}(\rho_{B|\Lambda^a_i}),
\end{align}
where the $\to$  indicates that the (ensemble-generating) measurement is done at party $A$ and quantum coherence is measured by party $B$.
%where $\lbrace M_i \rbrace_i$ for $i \in \lbrace 1, 2, 3 \rbrace$ denotes arbitrary mutually unbiased bases under which coherence is measured by Bob on conditional states of $\rho_{AB}$ whereas $\lbrace \Pi^a_i \rbrace^a$ represents any positive operator valued measurements (POVM) for particular coherence measurement denoted by $i$. This implies $\rho_{B|\Pi^a_i} = \frac{1}{p(\rho_{B|\Pi^a_i})} \lbrack (\Pi^a_i \otimes \mathbbm{1}_B) \rho_{AB} (\Pi^a_i \otimes \mathbbm{1}_B) \rbrack$ represents the conditional state of $\rho_{AB}$ for each $i$ with probability $p(\rho_{B|\Pi^a_i}) = \Tr [(\Pi^a_i \otimes \mathbbm{1}_B) \rho_{AB}]$ where $a$ denotes measurement outcome for each local measurement pereformed by Alice on her part and $\mathbbm{1}$ acts as identity operator.
We now formally state the standard criterion for detecting NAQC as follows.
\begin{criterion}\label{ghar1}[Standard criterion for NAQC detection]
Any two-qubit state $\rho_{AB}$ exhibits advantage in quantum coherence ``nonlocally" if
\begin{equation}\label{samna}
    \mathcal{N}^\rightarrow (\rho_{AB}) > \sqrt{6}.
\end{equation}
\end{criterion}
%The above criterion is sufficient but not a necessary one. 
We refer to $ \mathcal{N}^\rightarrow (\rho_{AB})$ as an NAQC functional. The above NAQC detection criterion is only a sufficient criterion, and not a necessary one, since Alice performs a measurement only in the bases contained in a set of MUBs, $\lbrace \Lambda^a_i \rbrace^a$, to detect NAQC. 
For the criterion to become necessary and sufficient, one needs to perform optimization over all possible measurements on Alice's side.% are included.

Next, in order to capture a larger set of states exhibiting NAQC, we consider a generalized version of the standard NAQC criterion.
In the generalized criterion, Alice is not restricted to performing measurements only in bases forming a set of MUBs. That is, she may choose a set of arbitrary three bases which, in general, may not form a set of MUBs, for measurement on her subsystem. 

\begin{criterion}\label{ghar2}[Generalized criterion for NAQC detection]
%$\mathsf{N}^\rightarrow (\rho_{AB})>\xi$ is the new nonlocal advantage of QC criterion for a two-qubit system to exhibit NAQC where
Any two-qubit state $\rho_{AB}$ shows advantage in quantum coherence ``nonlocally" if
\begin{equation}
\label{eq-naqc-max}
    \mathsf{N}^\rightarrow (\rho_{AB}) \coloneqq  \max_{\{M_i\}_i,\{\Pi_i\}_i} \sum_{i,a}  p(\rho_{B|\Pi^a_i}) C_{M_i}(\rho_{B|\Pi^a_i})>\sqrt{6},
\end{equation}
where $\{\{\Pi^a_i\}^a\}_i$ are a set of three arbitrary projective measurements, while $\lbrace M_i \rbrace_i$ are an arbitrary set of MUBs.
\end{criterion}

This criterion is potentially also sufficient and not necessary since not all measurements are spanned on Alice's side. But nevertheless, by construction, it is no less strong than the standard criterion,
%as all the states captured by criterion \ref{ghar1} will also be captured by the criterion \ref{ghar2}, 
i.e.,
\begin{equation}
    \mathsf{N}^\rightarrow (\rho_{AB}) \geq \mathcal{N}^\rightarrow (\rho_{AB}).
\end{equation}
We will explicitly show below that the generalized criterion can detect the nonlocal advantage of quantum coherence  of states which are not detected by the standard criterion.

Let us now demonstrate a useful property of invariance of the functionals, $\mathcal{N}^\rightarrow (\rho_{AB})$ and $\mathsf{N}^\rightarrow (\rho_{AB})$, under the action of local unitaries on the state in their arguments.
%Thus the generalized criterion guarantees that all the states, local unitarily connected to a state which exhibit NAQC, can also have NAQC, an intuitively satisfactory property.

\begin{theorem}
The NAQC functionals, $\mathcal{N}^\rightarrow (\rho_{AB})$ and $\mathsf{N}^\rightarrow (\rho_{AB})$,
%, generalized over projective measurements and MUBs, 
are invariant under the action of local unitaries for any quantum state, $\rho_{AB}$.
\end{theorem}
\begin{proof}
%First we will show that $p(\rho_{B|\Pi_i^a}) C_{M_i} (\rho_{B|\Pi_i^a})$ remains invariant under the transformation $\rho_{AB} \rightarrow (U \otimes \mathbbm{1}) \rho_{AB} (U^\dagger \otimes \mathbbm{1}) = \rho'_{AB}$. In the new definition of NAQC, there exist scope for transformations of measurements and therefore for the transformation, $\Pi_A^i \rightarrow U \Pi_A^i U^\dagger = \Pi_A^{'i}$ for any unitary $U$, where $j$ denotes the basis from a set of MUB, 
%Let $\rho$ is transformed into $(U \otimes \mathbbm{1}) \rho$ and. Since $\Tr_A{[(U \otimes \mathbbm{1}) \rho (U^\dagger \otimes \mathbbm{1})]} = \Tr_A[\rho]$ always true, so $\sum_i p_i C_j (\rho_{B|\Pi_A^i})$ is invariant for all $(U \otimes \mathbbm{1}) \rho$ where j belongs to a set of MUB.

Consider the transformation $\rho'_{AB} =(U \otimes V) \rho_{AB} (U^\dagger \otimes V^\dagger)$, where $U$ and $V$ are arbitrary unitary operators. 
We first want to show that $\mathcal{N}^\rightarrow (\rho_{AB})=\mathcal{N}^\rightarrow (\rho'_{AB})$.
Let $\{\{\Lambda_{i}^a\}^a\}_i$ 
%be the set of projective measurements forming a MUB 
and $\{M_i\}_i$ be sets of MUBs responsible for the maximization in \eqref{eq-naqc} for the state $\rho_{AB}$. Let $\rho_{B|\Lambda_i^a}$ be the conditional states at Bob's end, after Alice's measurements.

Now, consider the set of measurements, $\{\{\Lambda'^a_i\}^a\}_i=\{\{U \Lambda_i^a U^\dagger\}^a\}_i$, performed by Alice, and $\{M'_i\}_i=\{VM_iV^\dagger\}_i$, used by Bob on the state $\rho'_{AB}$ to obtain $\sum_{i,a}  p(\rho'_{B|\Lambda'^a_i}) C_{M'_i}(\rho'_{B|\Lambda'^a_i})$.
%\begin{eqnarray}
%  p_i C_j (\rho_{B|\Pi_A^i}) &\rightarrow& C_j (\Tr_A [\Pi_A^{'i}\otimes \mathbbm{1} \rho' \Pi_A^{'i}\otimes \mathbbm{1}]) \nonumber \\
%  &=& p_i C_j(\rho_{B|\Pi_A^i}).
%\end{eqnarray}
%Secondly, consider the transformation $\rho \to (\mathbbm{1} \otimes U)\rho (\mathbbm{1} \otimes U^\dagger)$ for any unitary $U$. 
Note that the reduced states at Bob's end, created due to measurements,  $\{\{\Lambda'^a_i\}^a\}_i$, by Alice on $\rho'_{AB}$ are $V \rho_{B|\Lambda_i^a} V^\dagger$ and its occurrence probability, 
$p(V \rho_{B|\Lambda_i^a} V^\dagger)=p(\rho_{B|\Lambda_i^a})$.
Also note that 
\begin{equation}
    C_{M_i} (\chi) = C_{V M_i V^\dagger}( V \chi  V^\dagger),
\end{equation}
where $\chi$ is a single-qubit density matrix.
Thus, for the measurements, $\{\{\Lambda'^a_i\}^a\}_i$ and $\{M'_i\}_i$, which are sets of MUBs as their parents were so,
%\begin{equation}
%    \sum_i p_i C_j (\rho_{B|\Pi_A^i}) = \sum_i p_i C_{UjU^\dagger} (U \rho_{B|\Pi_A^i} U^\dagger).
%\end{equation}
\begin{eqnarray}\label{zikr}
    \sum_{i,a}  p(\rho'_{B|\Lambda'^a_i}) C_{M'_i}(\rho'_{B|\Lambda'^a_i})&=&\mathcal{N}^\rightarrow (\rho_{AB}). \nonumber \\
    \implies \mathcal{N}^\rightarrow (\rho'_{AB}) &\geq& \mathcal{N}^\rightarrow (\rho_{AB}).
\end{eqnarray}
Now, $\rho_{AB}=(U^\dagger \otimes V^\dagger) \rho'_{AB} (U \otimes V)$, where $U^\dagger$ and $V^\dagger$ are also unitary, and thus a similar treatment as above leads to
\begin{equation} \label{zaika}
    \mathcal{N}^\rightarrow (\rho_{AB}) \geq \mathcal{N}^\rightarrow (\rho'_{AB}).
\end{equation}
Therefore, \eqref{zikr} and \eqref{zaika} prove that 
\begin{equation}\label{jagah}
    \mathcal{N}^\rightarrow (\rho'_{AB}) = \mathcal{N}^\rightarrow (\rho_{AB}),
\end{equation}
and thus we have proven invariance under local unitaries, of $\mathcal{N}^\rightarrow$.

Notice that the above proof can be repeated without requiring the measurements to belong to sets of MUBs, proving that the generalized NAQC functional, $\mathsf{N}^\rightarrow (\rho'_{AB})$, is also unaffected by local unitary action on the states.
\end{proof}

%We also demonstrate the invariance of the new NAQC under local unitaries. Following that, we'll present a relationship to make NAQC criterion simple to identify the states that exhibit NAQC and a specific example to show how our updated definition is more inclusive. Here, we only consider one measurement settings and projective measurements as local operations performed by Alice. As, NAQC is a sufficient criterion, we will measure $\mathcal{N}_1$ by maximizing over both projective measurements performed by Alice and coherence measurements on conditional states by Bob to detect more states get advantage of QC nonlocally. 

%\label{sec3}
%NAQC will be quantified by $\mathsf{N}^{\textnormal{max}}_1$ in the rest of our paper. In this section, we will check monogamy of NAQC in $\mc{C}^2 \otimes \mc{C}^2 \otimes \mc{C}^2$ tripartite system. 

Now we present a lemma providing a lower bound of the quantity $\mathcal{N}^\rightarrow (\rho_{AB})$, and thus a lower bound of the quantity $\mathsf{N}^\rightarrow (\rho_{AB})$.
\begin{lemma}
\label{lemma-naqc-lb}
The NAQC functional of an arbitrary state $\rho_{AB}$ 
%in $\mathbb{C}^2 \otimes \mathbb{C}^2$ system
of two qubits is lower bounded by the sum of the quantum coherences of the reduced state, $\Tr_A[\rho_{AB}]$ of $\rho_{AB}$, measured in an arbitrary set of mutually unbiased bases.
\end{lemma}

\begin{proof}
Let $\lbrace M_i \rbrace_i$ be an arbitrary set of mutually unbiased bases and
$\lbrace \Lambda^a_i \rbrace^a$ be another.
%an arbitrary set of MUBs. 
%with $\sum_a \Lambda^a_i = \mathbbm{1}$ for each $i$. 
Then, the  NAQC functional for any bipartite state $\rho_{AB}$ can be expressed as
\begin{widetext}
    \begin{align*}
    \mathcal{N}^\rightarrow (\rho_{AB}) &\coloneqq \max_{M_i,\Lambda^a_i} \sum_{i,a} p(\rho_{B|\Lambda^a_i}) C_{M_i} (\rho_{B|\Lambda^a_i})\nonumber\\ 
    &= \max_{M_i,\Lambda^a_i} \sum_{i,a} p(\rho_{B|\Lambda^a_i}) C_{M_i}\left( \Tr_A[\frac{\bra{\Lambda^a_i \otimes \mathbbm{1}_B} \rho_{AB} \ket{\Lambda^a_i \otimes \mathbbm{1}_B}}{p(\rho_{B|\Lambda^a_i})}]\right)\\
    &\geq \max_{M_i,\Lambda^a_i}  \sum_{i} C_{M_i} \left(\sum_{a}  p(\rho_{B|\Lambda^a_i}) \Tr_A[\frac{\bra{\Lambda^a_i \otimes \mathbbm{1}_B} \rho_{AB} \ket{\Lambda^a_i \otimes \mathbbm{1}_B}}{p(\rho_{B|\Lambda^a_i})}]\right)\\
    &= \max_{M_i} \sum_{i} C_{M_i}(\Tr_A[\rho_{AB}]) = \max_{M_i} \sum_{i} C_{M_i}(\rho_B), \\
    \implies \mathcal{N}^\rightarrow (\rho_{AB}) &\geq \max_{M_i} \sum_{i} C_{M_i}(\rho_B),
    \end{align*}
\end{widetext}
where the first inequality is due to the convexity of quantum coherence~\cite{l1-norm} and the following equality is due to the fact that any measurement on any subsystem of a system cannot disturb the average state of the other subsystems.

We have used the notation $\rho_B \coloneqq \Tr_A[\rho_{AB}]$ for the reduced density matrix of the state $\rho_{AB}$ in the B part. This completes the proof. 
%Similarly, it can be proved for generalized NAQC functional.
\end{proof}
%It is true for any coherence measure in any mutually unbiased basis and even it is valid for all POVMs measured by Alice. %Therefore, if $\mathcal{N}_1 (\rho_B) = \sqrt{6}$, then the state $\rho_{AB}$ will always equal or higher than the bound of NAQC.
%Therefore, if the sum of coherence of reduced density matrix of two-qubit state $\rho_{AB}$ in any mutually unbiased basis is $\sqrt{6}$, then $\mathcal{N}^\rightarrow (\rho_{AB})$ and thus $\mathsf{N}^\rightarrow (\rho_{AB})$ will definitely be greater than or equal to $\sqrt{6}$. Also if $\rho_{AB}$ is such that $\rho_B$ is maximally mixed then the lower bound to $\mathsf{N}^\rightarrow (\rho_{AB})$ becomes zero, as the coherence of maximally mixed state in any basis is zero.
%The above relation holds for new NAQC criterion as well.

%In the rest of our paper, a generalized projective operator $\lbrace \ket{\Pi_i^0}\bra{\Pi_i^0}, \ket{\Pi_i^1}\bra{\Pi_i^1} \rbrace$  is considered to perform optimization  over local measurements by Alice, with
%\begin{align} \ket{\Pi_i^0} &= \cos \frac{\theta_i}{2} \ket{0} + e^{\mathsf{i} \phi} \sin \frac{\theta_i}{2} \ket{1}, \label{eq-projective-0}\\\ket{\Pi^{1}} &= \sin \frac{\theta}{2} \ket{0} - e^{i \phi} \cos \frac{\theta}{2} \ket{1} \label{eq-projective-1}\end{align}
% for $\theta \in \lbrack 0, \pi \rbrack$, $\phi \in \lbrack 0, 2\pi \rbrack$, and $\lbrace \ket{0}, \ket{1} \rbrace$ is eigenvector of $\sigma_z$ Pauli matrix. 

Consider now an arbitrary set of mutually unbiased bases, $\lbrace M_i \rbrace_i$, for $i \in \lbrace 1,2,3 \rbrace$, of a single qubit, for which the elements of the bases can be expressed as follows:
\begin{eqnarray}
  \lbrace \ket{M_1^\pm} \rbrace &=& \lbrace \cos \frac{\theta'}{2} \ket{0} + e^{i \phi'} \sin \frac{\theta'}{2} \ket{1}, \nonumber \\ &&\sin \frac{\theta'}{2} \ket{0} - e^{i \phi'} \cos \frac{\theta'}{2} \ket{1} \rbrace, \label{eq-MUB-co-1}\\
   \ket{M_2^\pm} &=&  \frac{\ket{M_1^+} \pm \ket{M_1^-}}{\sqrt{2}} , \label{eq-MUB-co-2}\\
   \ket{M_3^\pm} &=&  \frac{\ket{M_1^+} \pm i \ket{M_1^-}}{\sqrt{2}}, \label{eq-MUB-co-3}
\end{eqnarray}
where $ \ket{0}, \ket{1} $ are eigenvectors of the $\sigma_z$ matrix, and $\theta'$, $\phi'$ are azimuthal and polar angles in spherical polar coordinates. Quantum coherences of a general single-qubit state in these bases are obtained in Appendix \ref{appendix-MUB-qubit}. These are important for evaluation of the NAQC functionals, later in the paper.
\begin{figure}[h!]
\includegraphics[scale=1.0]{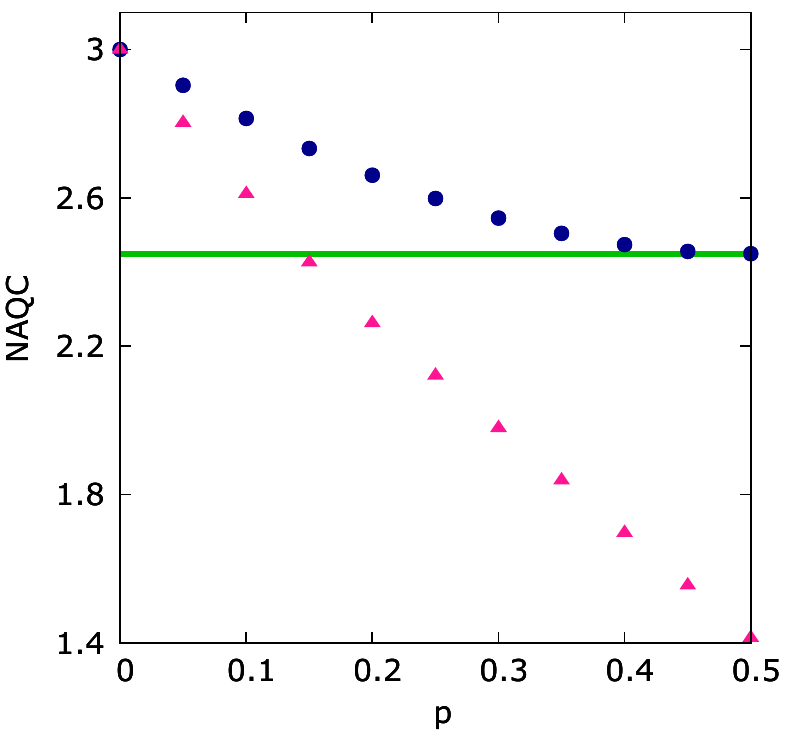}
\caption{(Color online)
NAQC for probabilistic mixtures of two Bell states. We plot the NAQC functionals (vertical axis) for mixtures of two Bell states, parameterized by the mixing parameter $p$ (horizontal axis). Pink triangles and blue circles denote the standard and generalized NAQC funtionals, respectively, whereas the green line represents the upper bound beyond which NAQC occurs in a two-qubit system. We find that the states with $p \lesssim 0.144$ exhibit NAQC using the standard criterion, whereas all states with $p \lesssim 0.5$ exhibit NAQC using the generalized criterion. This shows that the generalized criterion detects more two-qubit states exhibiting NAQC than the standard criterion. All quantities used are dimensionless.}
\label{fig1} 
\end{figure} 

It has been realized that non-steerable bipartite states can never exhibit nonlocal advantage of quantum coherence~\cite{debasis}, and thus the same follows for the standard NAQC as well as the generalized NAQC functionals (see Appendix \ref{appendix-naqc-steering}). Therefore, 
%in particular, 
separable states will never exhibit NAQC, as indicated by using the standard NAQC and generalized NAQC functionals~\cite{pati}.
%when NAQC and generalized NAQC functionals are taken into account to detect NAQC.

\begin{figure}[h!]
\includegraphics[scale=1.0]{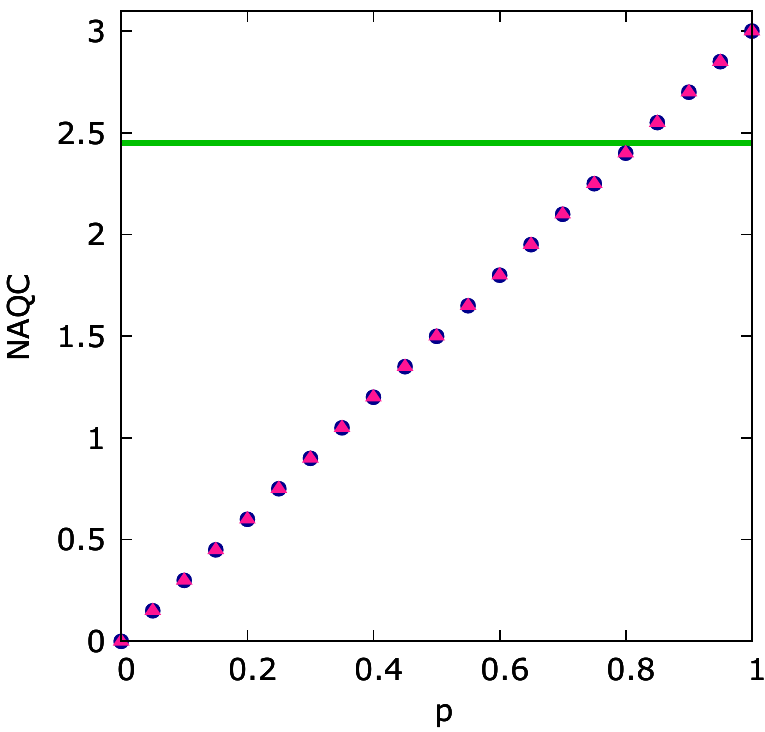}
\caption{(Color online)
Nature of NAQC for two-qubit Werner states. All considerations are the same as in Fig. \ref{fig1} except that here the states under discussion are two-qubit Werner states. They exhibit NAQC for the mixing parameter, 
%and the corresponding states with 
$p \gtrsim 0.815$, using both the NAQC functionals considered in this paper.}
\label{fig2} 
\end{figure} 

A classical (probabilistic) mixture of two Bell states, say $\ket{\phi^{+}} ( \coloneqq \frac{1}{\sqrt{2}} \lbrack \ket{00} + \ket{11}\rbrack )$ and $\ket{\psi^{+}} (\coloneqq \frac{1}{\sqrt{2}} \lbrack \ket{10} + \ket{01}\rbrack ) $, can be represented as
\begin{align}
\label{eq-state-pure}
    \rho_{AB} = p \ket{\phi^+} \bra{\phi^+} + (1-p) \ket{\psi^+} \bra{\psi^+},
\end{align}
with $p \in \lbrack 0, 1 \rbrack$. 
%where $\ket{\phi^{+}} \coloneqq \frac{1}{\sqrt{2}} (\ket{00} + \ket{11})$ and $\ket{\psi^{+}} = \frac{1}{\sqrt{2}} (\ket{10} + \ket{01})$. 
Using this set of states, we find that the generalized NAQC criterion detects strictly more two-qubit states exhibiting NAQC than the standard NAQC criterion. This is shown in Fig. \ref{fig1}. 
As another example, consider the two-qubit Werner states,
\begin{equation}
    \rho_W = p \ket{\phi^+} \bra{\phi^+} + \frac{(1-p)}{4} \mathbbm{1}, 
\end{equation}
with $p \in \lbrack 0 , 1\rbrack$ and $\mathbbm{1}$ representing the identity operator on the $\mathbb{C}^2 \otimes \mathbb{C}^2$ Hilbert space. It is known that two-qubit Werner states are entangled if $p > \frac{1}{3}$~\cite{peres-PPT,horodecki-PPT} and steerable if $p > \frac{1}{2}$~\cite{werner-different-range}. Interestingly, we observe in Fig. \ref{fig2} that there exist entangled and steerable states which does not exhibit nonlocal advantage of quantum coherence when the standard NAQC or the generalized NAQC functionals are taken into account. So, the set of states showing NAQC, using both the NAQC functionals, is still a strict subset of the set of entangled states, as well as that of the set of steerable states.

\section{Monogamy of NAQC}
\label{sec4}

%I have cited the papers here from monogamy of disocrd paper of sir.
%Here, we will discuss about monogamy of nonlocal advantage of quantum coherence (NAQC) for any pthree-qubit system. 

%We will symbolize NAQC functional for state $\rho_{AB}$ by $\mathsf{N}^\leftarrow {(\rho_{AB})}$ when coherence in any MUBs, will be measured on the reduced states of $\rho_{AB}$ after the projective measurements on subsystem B whereas $\mathsf{N}^\rightarrow {(\rho_{AB})}$ denotes the opposite case of previous one.
Since a NAQC functional of a bipartite system involves quantum measurements on one of the subsystems and quantum coherence of reduced states on the other subsystem, monogamy of NAQC for tripartite state $\rho_{ABC}$ can be seen from two different perspectives. In one of them, quantum coherence is measured on a fixed subsystem, $A$, while measurements for creating ensembles at A are performed on one of the remaining
%is performed by switching the remainder 
subsystems, viz. $B$ and $C$, so that the quantity under study is
\begin{equation}
\label{same-co-eqn}
    \mathsf{N}^\leftarrow(\rho_{AB}) + \mathsf{N}^\leftarrow(\rho_{AC}).
\end{equation}
In the other one, monogamy of NAQC is studied in the situation where quantum measurements for generating ensembles are performed on a fixed subsystem $A$ of the multipartite system, and quantum coherence of states is measured by switching the remaining subsystems, i.e., the quantity studied is
\begin{eqnarray}
\label{same-mea-eqn}
  \mathsf{N}^\rightarrow(\rho_{AB}) +\mathsf{N}^\rightarrow(\rho_{AC}).
\end{eqnarray}
Here,
%$\rho_{A_0A_1A_2...A_{\textnormal{N}}}$ is $\textnormal{(N+1)}$-partite system acting on $\mathcal{H}_1 \otimes \mathcal{H}_2 \otimes ... \otimes \mathcal{H}_{N+1}$ and
~$\rho_{AB}$ and $\rho_{AC}$ denote reduced states of $\rho_{ABC}$ by tracing out parties $C$ and $B$, respectively. We now note that $2\sqrt{6}$ would be a non-trivial upper bound for the above sums, for monogamy of NAQC in the tripartite system ABC. %given few two-qubit pure states can show NAQC.
% Here, the subsystem $A_0$ is fixed on which coherence is measured for the first inequality whereas in the second inequality quantum measurement is performed on the fixed party $A_0$. 
The inequalities would then suggest that if the pair $AB$ shows NAQC, then the pair $AC$ will not exhibit NAQC, and vice-versa. 
%Note that the generalizedvalue of NAQC functional over two qubit pure states is 3.011 which is greater than $\sqrt{6}$. Thus, the above inequalities are stronger forms of monogamy for NAQC. 
Such monogamy has previously been observed for Bell correlations~\cite{toner-bell-2,toner-bell-1,kazlikowski-bell-monogamy-1,kazlikowski-bell-monogamy-2}, and for quantum dense coding~\cite{prabhu-dense-coding}.

\begin{figure}[h!]
\includegraphics[scale=0.5]{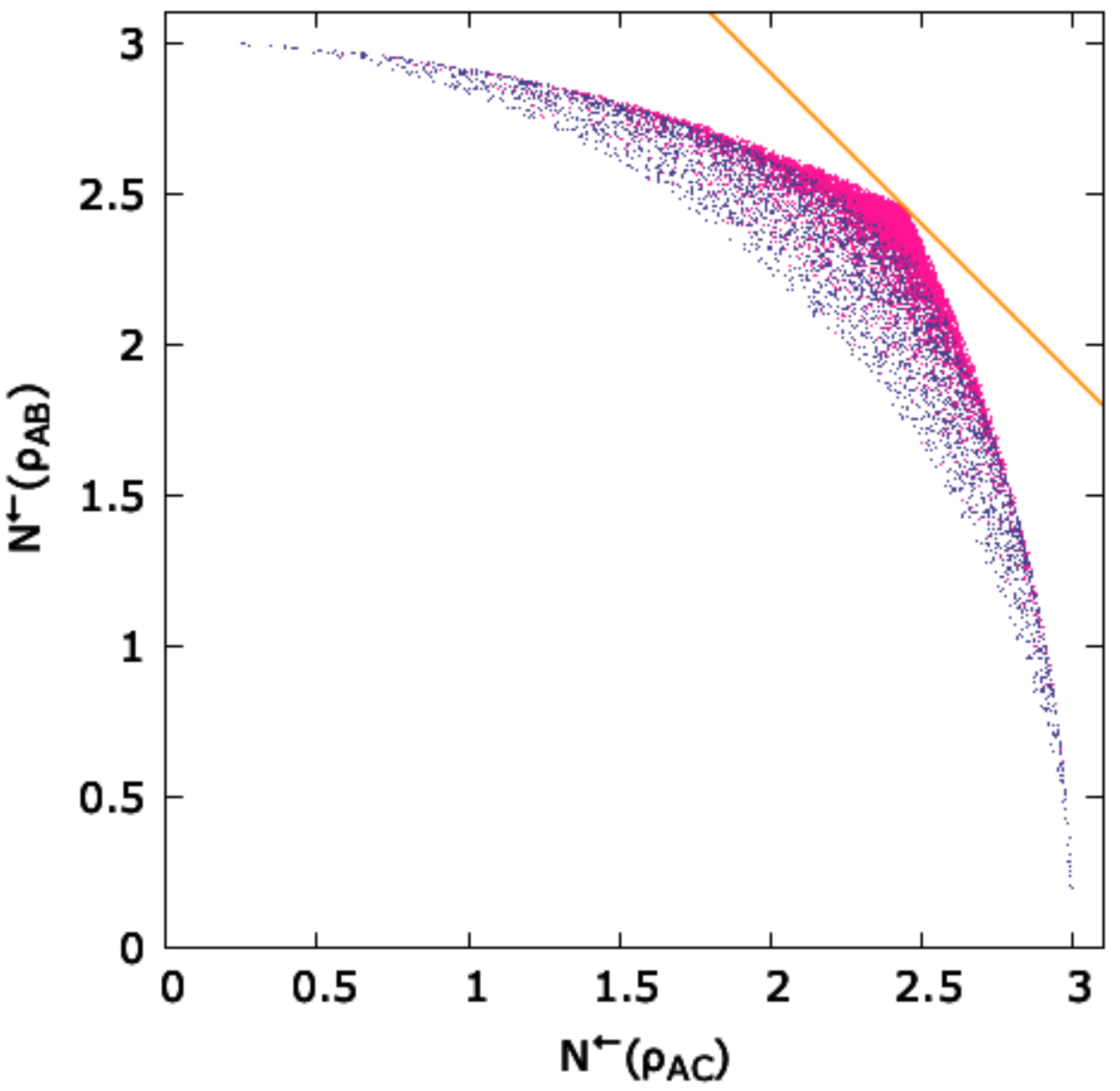}
\caption{Exclusion principle for NAQC in three-qubit system. We analyze here the nature of monogamy of the generalized NAQC functional when the reduced states of the fixed subsystem $A$ are subjected to coherence measurement for any pure three-qubit state of $ABC$. We plot $\mathsf{N}^\leftarrow(\rho_{AC})$ along the horizontal axis and $\mathsf{N}^\leftarrow(\rho_{AB})$ along the vertical axis. GHZ- and W-class
%, bi-separable, and fully separable 
genuine three-qubit entangled pure states are represented by pink and dark blue circles, respectively, while the condition $\mathsf{N}^\leftarrow(\rho_{AC}) + \mathsf{N}^\leftarrow(\rho_{AB}) = 2\sqrt{6}$ is denoted by the orange straight line. All quantities used are dimensionless.}
\label{fig3} 
\end{figure}
%bi-separable=green, red=fully separable, pink = ghz, dark blue = w

\subsection{Monogamy when coherence is measured on a fixed subsystem}
We examine here monogamy of the generalized NAQC functional for three-qubit pure states $\rho_{ABC}$, 
%acting on $\mathbb{C}^2 \otimes \mathbb{C}^2 \otimes \mathbb{C}^2$ Hilbert space 
where quantum coherence is measured on the fixed subsystem, $A$. We therefore focus on the expression in \eqref{same-co-eqn}.

Note that the upper bound of \eqref{same-co-eqn} that can be achieved, by fully separable pure three-qubit states, is $2\sqrt{6}$. For all bi-separable pure three-qubit states, the maximum of $\mathsf{N}^\leftarrow(\rho_{AB}) + \mathsf{N}^\leftarrow(\rho_{AC})$ is also $(2\sqrt{6})$.
By employing a non-linear optimization routine, we numerically analyze monogamy of the generalized NAQC functional for all GHZ- and W-class three-qubit pure states. We find that the expression in \eqref{same-co-eqn} can reach a maximum of $\mathsf{N}^\leftarrow(\rho_{AB}) + \mathsf{N}^\leftarrow(\rho_{AC}) \approx 4.899 \approx 2\sqrt{6}$. Since the upper bound of $\mathsf{N}^\leftarrow(\rho_{AB}) + \mathsf{N}^\leftarrow(\rho_{AC})$ for all three-qubit pure states is $2\sqrt{6}$, we therefore have a strong monogamy - an exclusion principle - for such states, as if $\rho_{AB}$ exhibits NAQC, $\rho_{AC}$ does not, and vice versa, where the quantum coherence measurements in both cases are at the party A.
%s is strictly conserved, when coherence is measured on same party of reduced states of the system.
%, i.e., \begin{equation}    \mathsf{N}^\leftarrow(\rho_{AC}) + \mathsf{N}^\leftarrow(\rho_{AB}) \leq 2\sqrt{6},\end{equation}holds for all three qubit pure states.
In Fig. \ref{fig3}, we also plot a scatter diagram of $\mathsf{N}^\leftarrow(\rho_{AB})$ versus $\mathsf{N}^\leftarrow(\rho_{AC})$ for Haar-uniformly generated 
%GHZ- and W-class three-qubit pure states.
%
Greenberger–Horne–Zeilinger- (GHZ-)~\cite{ghz-class-1,ghz-class-2} and W-class~\cite{w-class-1,three-pure-2} three-qubit pure states.
%
%by generating  Haar-uniform three qubit pure states 
%(upto a local unitary). 
%which is depicted in the Fig. \ref{fig3}. It can be observed that Haar uniformly generated three-qubit pure states are monogamous with respect to the NAQC functional for the case when coherence of states is measured on a fixed subsystem of the three-qubit system.
%by optimizing Eq. \eqref{eq-naqc-co-same} over all $\ket{M_i}$ and $\ket{\Pi^i}$ 

\begin{figure}[h!]
\includegraphics[scale=0.5]{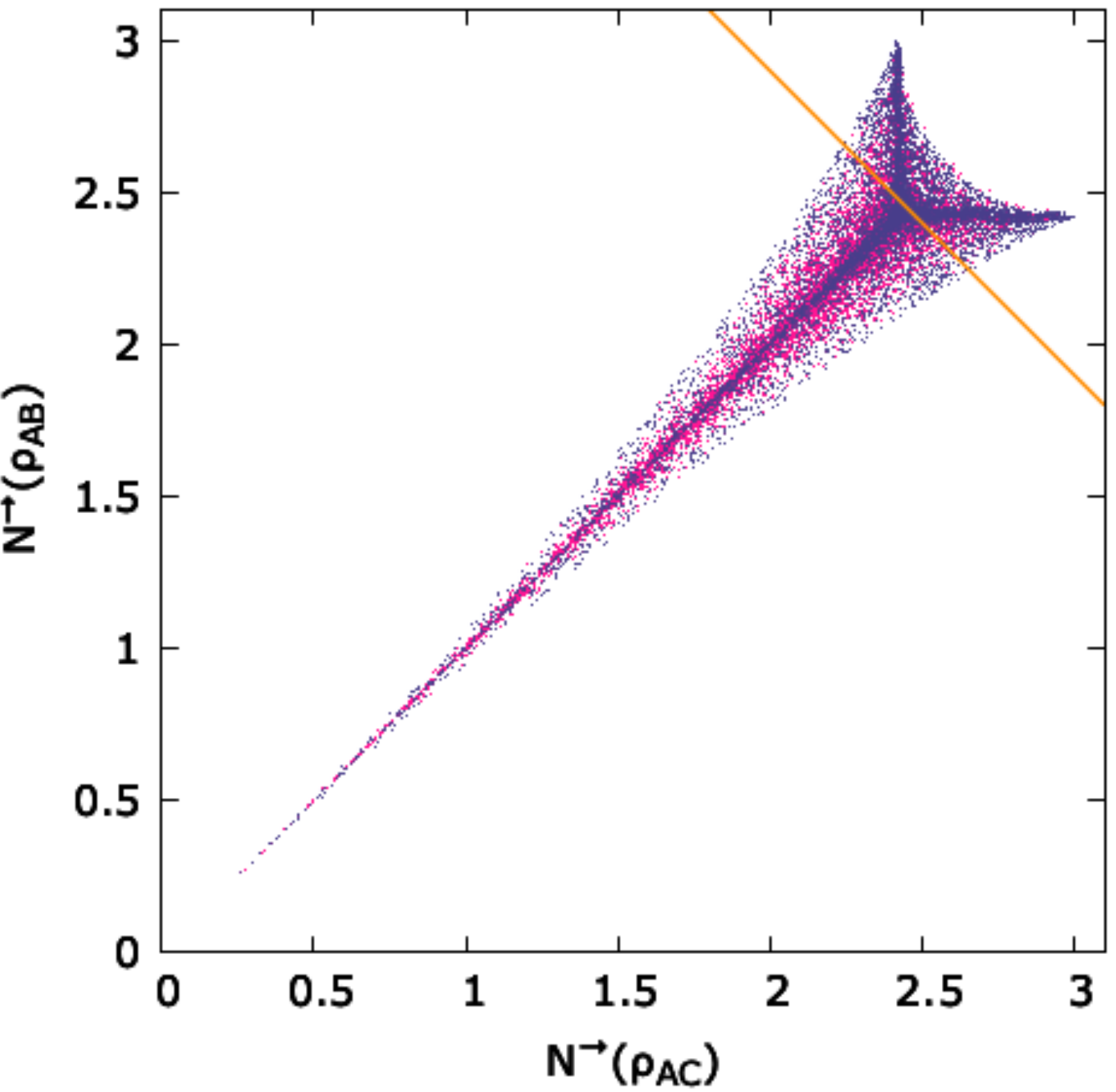}
\caption{%Witness for entanglement in a bipartition in a tripartite scenario. 
Sharing of nonlocal advantage of quantum coherence, when ensemble-generating measurements are performed on a fixed party. We plot $\mathsf{N}^\rightarrow(\rho_{AB})$ on the vertical axis against $\mathsf{N}^\rightarrow(\rho_{AC})$ represented on the horizontal axis, for Haar-uniformly generated pure genuine three-qubit entangled states, viz. the states of the GHZ- and W-classes. GHZ- and W-class states are denoted by pink and dark blue points, respectively, whereas the orange line represents the equation of the straight line, $\mathsf{N}^\rightarrow(\rho_{AC}) + \mathsf{N}^\rightarrow(\rho_{AB}) = 2\sqrt{6}$. 
%The generalized NAQC of reduced states $\rho_{AC}$ and $\rho_{AB}$ of any pure genuine three-qubit entangled states respectively, $\mathsf{N}^\rightarrow(\rho_{AC})$ and $\mathsf{N}^\rightarrow(\rho_{AB})$ are dimensionless.
All quantities used are dimensionless.}
\label{fig4} 
\end{figure}
%ghz=pink, w= dark blue

\subsection{When ensemble-preparing measurement is performed on a fixed subsystem}
Finally, we analyze monogamy properties of the generalized NAQC functional for any three-qubit pure state under the assumption that the ensemble-generating projective measurements are performed on a fixed subsystem of the tripartite system.
%and see if the monogamy of NAQC holds true in this circumstance. %The monogamy relation of generalized NAQC functional for this scenario will have
%\begin{align*}
%\mathsf{N}^\rightarrow(\rho_{AC}) + \mathsf{N}^\rightarrow(\rho_{AB}) \leq 2\sqrt{6},    
%\end{align*}
% with reduced states of a three-qubit pure state $\rho_{ABC}$, $\rho_{AB} (\coloneqq \Tr_C{[\rho_{ABC}}])$ and $\rho_{AC} (\coloneqq \Tr_B{[\rho_{ABC}]})$.

%In this scenario, similarly $\rho_{AC}$ and $\rho_{AB}$ from general pure three-qubit state $\ket{\Psi}_{ABC}$ are constructed (look in Appendix \ref{appendix-reduced-state}). Following arbitrary rank-1 projective measurement on party $A$ in both situations, the conditional states on parties $C$ and $B$ are then discovered, and the coherence of the conditional states received on parties $C$ and $B$ is measured using Eqs. \eqref{eq-co-0}, \eqref{eq-co-1}, and \eqref{eq-co-2}. Then, we'll calculate $\mathsf{N}_1$ numerically for the bipartite states $\rho_{AB}$ and $\rho_{AC}$ , and 

We notice that the maximum of \eqref{same-mea-eqn}, for all fully separable and bi-separable three-qubit pure states, would be $2\sqrt{6}$ and $(3+\sqrt{6})$ respectively. We also observe that, unlike the previous case, a generic pure three-qubit state may not always satisfy the strong monogamy - the exclusion principle - for the generalized NAQC, when the ensemble-generating projective measurements are performed on a fixed subsystem. I.e., there exists genuine three-party entangled pure states that violate the inequality, $\mathsf{N}^\rightarrow(\rho_{AC}) + \mathsf{N}^\rightarrow(\rho_{AB}) \leq 2\sqrt{6}$.
%by a three-qubit state, $\rho_{ABC}$.
%of the corresponding system and
This is demonstrated in Fig. \ref{fig4} by plotting $\mathsf{N}^\rightarrow(\rho_{AB})$ with respect to $\mathsf{N}^\rightarrow(\rho_{AC})$ for Haar-uniformly generated pure three-qubit states.
%few number of bi-separable, GHZ-class, and W-class states violate the monogamy of NAQC when projective measurement takes place on the same party of the tripartite system, 
%Furthermore, it is found numerically that all fully separable pure three-qubit states and the pure three-qubit states in which there is no entanglement present in between the bipartition of A and BC of the state $\rho_{ABC}$, are always monogamous for generalized NAQC in this scenario. Therefore, we can conclude that for any pure three-qubit state, for which the ensemble preparing measurement is performed on a fixed party, say A, the state
%tripartite system that performs projective measurement on the same party of the reduced states of the tripartite states, acting on $\mathbb{C}^2 \otimes \mathbb{C}^2 \otimes \mathbb{C}^2$, 
%must be entangled between A and BC in order to violate the monogamy of NAQC.

\section{conclusion}
\label{sec5}
%The nonlocal advantage of quantum coherence is a suitable criterion to examine how the measurement of one party in a bipartite system impacts in the quantum coherence of the other party at a finite distance from the former party of the corresponding bipartite composite system. In other words, it shows that a two-qubit system exhibits the advantage in quantum coherence nonlocally if the NAQC functional achieves the value greater than the bound of coherence complementarity relations of single-qubit system.

%Composite quantum systems possess various correlations like entanglement, steerable, and Bell nonlocal.

Nonlocal advantage of quantum coherence captures a form of ``steerability" of bipartite quantum states in terms of quantum coherence. There exists a quantum coherence complementarity relation for an isolated single-qubit system, stating that the sum of quantum coherences in any set of mutually unbiased bases is non-trivially bounded from above. A two-qubit state can exhibit NAQC if for some measurements on one of the subsystems, the sum of the average quantum coherences of the other subsystem violates the isolated single-qubit complementarity relation.

In this paper, we have considered detection of NAQC by using two criteria. The first one is referred to as the standard criterion, which allows measurements in arbitrary set of MUBs on one of the subsystems of the bipartite system, for creating ensembles. The second criterion is termed the generalized criterion, as the restriction of the ensemble-generating measurement bases to be a set of MUBs is relaxed. 
In both the criteria, an optimization of the NAQC functionals over all the relevant measurement bases is considered, which makes both the NAQC functionals local unitarily invariant over the states. We also explicitly demonstrated that the generalized NAQC criterion performs better than the standard NAQC criterion, in that the former can capture a greater number of bipartite states exhibiting nonlocal advantage of quantum coherence. In addition, we provided a lower bound on both the NAQC functionals for bipartite systems in terms of the quantum coherence of a reduced state.
%of the corresponding system.

Lastly, we examine the monogamy of NAQC in tripartite systems in two different cases - first, when coherence of states of a fixed subsystem of the tripartite system is measured, and second, when the ensemble-generating projective measurements are carried out on a fixed subsystem of the tripartite system. 
%In the latter scenario, we have only limited our measurements upto rank-1 projective measurements.
%We solely examine the monogamy of NAQC of three qubit pure quantum states.
In the first case, it is shown that all three-qubit pure states follows a strong monogamy - an exclusion principle - of NAQC. However, in the second case, the strong monogamy of NAQC may or may not be followed for general pure three-qubit states. %Further, it is found that if projective measurement is carried out on a fixed party $A$ of the tripartite system $ABC$, any fully separable three-qubit states and bi-separable states that do not have any entanglement between $A$ and $BC$ partition would never violate monogamy of generalized NAQC. 
\section*{Acknowledgment}
We acknowledge partial support from the Department of Science and Technology, Government of India through the QuEST grant (grant number DST/ICPS/QUST/Theme3/2019/120).

\begin{appendix}
\section{Quantum coherence of a qubit}
\label{appendix-MUB-qubit} 
We provide here that the quantum coherence of a qubit state with respect to a basis chosen from an arbitrary set of MUBs on the qubit Hilbert space. A single-qubit state, $\rho$, expressed in the $\sigma_z$ basis, will have the form
\begin{align*}
    \rho &\coloneqq \left( \begin{array}{cc}
\rho_{00} & \rho_{10}\\
\rho_{01} & \rho_{11} \end{array} \right). 
\end{align*}
%where $\rho_{00}, \rho_{01}, \rho_{10}, \rho_{11}$ are the components of $\rho$ written in $\lbrace \ket{0}, \ket{1} \rbrace$ basis.

The $l_1$-norm of quantum coherence with respect to the 
%$\lbrace M_1^\pm, M_2^\pm, M_3^\pm \rbrace$ 
bases in Eqs. \eqref{eq-MUB-co-1}, \eqref{eq-MUB-co-2}, and \eqref{eq-MUB-co-3} will be
\begin{align}
 C_{M_1^\pm} &= 2 |\frac{1}{2} \sin \theta' \rho_{00} - \frac{1}{2} \sin \theta' \rho_{11} - e^{i \phi'} \cos^2 \frac{\theta'}{2} \rho_{01} \nonumber \\ &+ e^{-i \phi'} \sin^2 \frac{\theta'}{2} \rho_{10}|, \label{eq-co-0} \\ 
 C_{M_2^\pm} &= |\cos \theta' \rho_{00} - \cos \theta' \rho_{11} + e^{i \phi'} (1+ \sin \theta') \rho_{01} \nonumber \\ & - e^{-i \phi'} (1- \sin \theta') \rho_{10}|,\label{eq-co-1} \\  \label{eq-co-2}
    C_{M_3^\pm} &= |\rho_{00} - \rho_{11} + i e^{i \phi'} \rho_{01} + i e^{-i \phi'} \rho_{10}|.
\end{align}
%where $\theta'$ and $\phi'$ are polar and azimuthal angle of arbitrary mutually unbiased bases considered in the paper.

\section{Proof of non-exhibition of generalized NAQC by any non-steerable two-qubit state}
\label{appendix-naqc-steering}
In quantum information tasks, steerability~\cite{steer-rev} is a significant quantum resource.
Let us assume that Alice and Bob are situated in two distant labs, and share a state, $\rho_{AB}$, between them. Suppose that Alice makes a measurement in the setting ``x", and obtains the outcome ``a''. Let the conditional state obtained thereby at Bob be $\rho(a|x)$. Assume also that there exists a hidden variable $\lambda$, distributed as $p_{\lambda}$, such that there exists a ``hidden" state, $\sigma_B(\lambda)$ at Bob and a conditional probability $p(a|x,\lambda)$ with 
%According to local hidden state model (LHS), after Alice local operation $x$ on her end of the composite system, if Bob's state can not be written as follows
\begin{equation*}
    \rho (a|x) = \sum_\lambda  p_\lambda p(a|x,\lambda) \sigma_B (\lambda).
\end{equation*}
% where $\rho (a|x)$ is conditional state on Bob side with probability $p(a|x,\lambda)$ and $\lambda$ acts as a hidden variable while $\sigma_\lambda$ is hidden state with probability $p(\lambda)$ satisfying $\sum_\lambda p(\lambda) = 1$, then the state shared between Alice and Bob is definitely a steered state. 
 %Steeribility lies between entanglement and Bell nonlocality for bipartite states. 
Then we call the state $\rho_{AB}$ as non-steerable if the above relation is available for every setting $x$ and every outcome $a$. Otherwise, it is steerable.

The NAQC for any non-steerable bipartite state can be determined as
\begin{align*}
  &\mathcal{N}^\rightarrow (\rho_{AB}) \\&\coloneqq \max_{M_i,\Lambda^a_i} \sum_{i,a} p(\rho_{B|\Lambda^a_i}) C_{M_i} (\rho_{B|\Lambda^a_i})\\
  &= \max_{M_i,\Lambda^a_i} \sum_{i,a} p(\rho_{B|\Lambda^a_i}) C_{M_i} \biggl( \frac{\sum_\lambda p_\lambda p(a|\Lambda^a_i, \lambda) \sigma_B (\lambda)}{p(\rho_{B|\Lambda^a_i})} \biggr)\\ 
  &\leq \max_{M_i,\Lambda^a_i} \sum_{i,a,\lambda} p_\lambda p(a|\Lambda^a_i, \lambda) C_{M_i} (\sigma_B (\lambda)),\\
  &= \max_{M_i} \sum_{i,\lambda} p_\lambda C_{M_i} (\sigma_B (\lambda))\\
  &\leq \sum_{i,\lambda} p_\lambda \sqrt{6} = \sqrt{6},
\end{align*}
where the first and second inequalities are obtained respectively, by utilizing the convexity of quantum coherence and the coherence complementarity relation of single-qubit systems.

Hence, any non-steerable two-qubit state will never exhibit NAQC. Analogously, it can be proved for the generalized NAQC functional.
\end{appendix}

\bibliography{naqc_qm}

\begin{thebibliography}{57}%
\makeatletter
\providecommand \@ifxundefined [1]{%
 \@ifx{#1\undefined}
}%
\providecommand \@ifnum [1]{%
 \ifnum #1\expandafter \@firstoftwo
 \else \expandafter \@secondoftwo
 \fi
}%
\providecommand \@ifx [1]{%
 \ifx #1\expandafter \@firstoftwo
 \else \expandafter \@secondoftwo
 \fi
}%
\providecommand \natexlab [1]{#1}%
\providecommand \enquote  [1]{``#1''}%
\providecommand \bibnamefont  [1]{#1}%
\providecommand \bibfnamefont [1]{#1}%
\providecommand \citenamefont [1]{#1}%
\providecommand \href@noop [0]{\@secondoftwo}%
\providecommand \href [0]{\begingroup \@sanitize@url \@href}%
\providecommand \@href[1]{\@@startlink{#1}\@@href}%
\providecommand \@@href[1]{\endgroup#1\@@endlink}%
\providecommand \@sanitize@url [0]{\catcode `\\12\catcode `\$12\catcode
  `\&12\catcode `\#12\catcode `\^12\catcode `\_12\catcode `\%12\relax}%
\providecommand \@@startlink[1]{}%
\providecommand \@@endlink[0]{}%
\providecommand \url  [0]{\begingroup\@sanitize@url \@url }%
\providecommand \@url [1]{\endgroup\@href {#1}{\urlprefix }}%
\providecommand \urlprefix  [0]{URL }%
\providecommand \Eprint [0]{\href }%
\providecommand \doibase [0]{http://dx.doi.org/}%
\providecommand \selectlanguage [0]{\@gobble}%
\providecommand \bibinfo  [0]{\@secondoftwo}%
\providecommand \bibfield  [0]{\@secondoftwo}%
\providecommand \translation [1]{[#1]}%
\providecommand \BibitemOpen [0]{}%
\providecommand \bibitemStop [0]{}%
\providecommand \bibitemNoStop [0]{.\EOS\space}%
\providecommand \EOS [0]{\spacefactor3000\relax}%
\providecommand \BibitemShut  [1]{\csname bibitem#1\endcsname}%
\let\auto@bib@innerbib\@empty
%</preamble>
\bibitem [{\citenamefont {Streltsov}\ \emph {et~al.}(2017)\citenamefont
  {Streltsov}, \citenamefont {Adesso},\ and\ \citenamefont
  {Plenio}}]{review-coherence-1}%
  \BibitemOpen
  \bibfield  {author} {\bibinfo {author} {\bibfnamefont {A.}~\bibnamefont
  {Streltsov}}, \bibinfo {author} {\bibfnamefont {G.}~\bibnamefont {Adesso}}, \
  and\ \bibinfo {author} {\bibfnamefont {M.~B.}\ \bibnamefont {Plenio}},\
  }\bibfield  {title} {\enquote {\bibinfo {title} {Colloquium: Quantum
  coherence as a resource},}\ }\href
  {https://link.aps.org/doi/10.1103/RevModPhys.89.041003} {\bibfield  {journal}
  {\bibinfo  {journal} {Rev. Mod. Phys.}\ }\textbf {\bibinfo {volume} {89}},\
  \bibinfo {pages} {041003} (\bibinfo {year} {2017})}\BibitemShut {NoStop}%
\bibitem [{\citenamefont {Giovannetti}\ \emph {et~al.}(2011)\citenamefont
  {Giovannetti}, \citenamefont {Lloyd},\ and\ \citenamefont
  {Maccone}}]{Giovannetti_2011}%
  \BibitemOpen
  \bibfield  {author} {\bibinfo {author} {\bibfnamefont {V.}~\bibnamefont
  {Giovannetti}}, \bibinfo {author} {\bibfnamefont {S.}~\bibnamefont {Lloyd}},
  \ and\ \bibinfo {author} {\bibfnamefont {L.}~\bibnamefont {Maccone}},\
  }\bibfield  {title} {\enquote {\bibinfo {title} {Advances in quantum
  metrology},}\ }\href {\doibase 10.1038/nphoton.2011.35} {\bibfield  {journal}
  {\bibinfo  {journal} {Nature Photonics}\ }\textbf {\bibinfo {volume} {5}},\
  \bibinfo {pages} {222} (\bibinfo {year} {2011})}\BibitemShut {NoStop}%
\bibitem [{\citenamefont {Giorda}\ and\ \citenamefont
  {Allegra}(2017)}]{Giorda18}%
  \BibitemOpen
  \bibfield  {author} {\bibinfo {author} {\bibfnamefont {P.}~\bibnamefont
  {Giorda}}\ and\ \bibinfo {author} {\bibfnamefont {M.}~\bibnamefont
  {Allegra}},\ }\bibfield  {title} {\enquote {\bibinfo {title} {Coherence in
  quantum estimation},}\ }\href {\doibase 10.1088/1751-8121/aa9808} {\bibfield
  {journal} {\bibinfo  {journal} {Journal of Physics A: Mathematical and
  Theoretical}\ }\textbf {\bibinfo {volume} {51}},\ \bibinfo {pages} {025302}
  (\bibinfo {year} {2017})}\BibitemShut {NoStop}%
\bibitem [{\citenamefont {Plenio}\ and\ \citenamefont
  {Huelga}(2008)}]{Plenio_2008}%
  \BibitemOpen
  \bibfield  {author} {\bibinfo {author} {\bibfnamefont {M.~B.}\ \bibnamefont
  {Plenio}}\ and\ \bibinfo {author} {\bibfnamefont {S.~F.}\ \bibnamefont
  {Huelga}},\ }\bibfield  {title} {\enquote {\bibinfo {title}
  {Dephasing-assisted transport: quantum networks and biomolecules},}\ }\href
  {https://dx.doi.org/10.1088/1367-2630/10/11/113019} {\bibfield  {journal}
  {\bibinfo  {journal} {New Journal of Physics}\ }\textbf {\bibinfo {volume}
  {10}},\ \bibinfo {pages} {113019} (\bibinfo {year} {2008})}\BibitemShut
  {NoStop}%
\bibitem [{\citenamefont {Rebentrost}\ \emph {et~al.}(2009)\citenamefont
  {Rebentrost}, \citenamefont {Mohseni},\ and\ \citenamefont
  {Aspuru-Guzik}}]{Rebentrost_2009}%
  \BibitemOpen
  \bibfield  {author} {\bibinfo {author} {\bibfnamefont {P.}~\bibnamefont
  {Rebentrost}}, \bibinfo {author} {\bibfnamefont {M.}~\bibnamefont {Mohseni}},
  \ and\ \bibinfo {author} {\bibfnamefont {A.}~\bibnamefont {Aspuru-Guzik}},\
  }\bibfield  {title} {\enquote {\bibinfo {title} {Role of quantum coherence
  and environmental fluctuations in chromophoric energy transport},}\ }\href
  {https://doi.org/10.1021%2Fjp901724d} {\bibfield  {journal} {\bibinfo
  {journal} {The Journal of Physical Chemistry B}\ }\textbf {\bibinfo {volume}
  {113}},\ \bibinfo {pages} {9942} (\bibinfo {year} {2009})}\BibitemShut
  {NoStop}%
\bibitem [{\citenamefont {Lloyd}(2011)}]{Lloyd_2011}%
  \BibitemOpen
  \bibfield  {author} {\bibinfo {author} {\bibfnamefont {S.}~\bibnamefont
  {Lloyd}},\ }\bibfield  {title} {\enquote {\bibinfo {title} {Quantum coherence
  in biological systems},}\ }\href
  {https://dx.doi.org/10.1088/1742-6596/302/1/012037} {\bibfield  {journal}
  {\bibinfo  {journal} {Journal of Physics: Conference Series}\ }\textbf
  {\bibinfo {volume} {302}},\ \bibinfo {pages} {012037} (\bibinfo {year}
  {2011})}\BibitemShut {NoStop}%
\bibitem [{\citenamefont {Huelga}\ and\ \citenamefont
  {Plenio}(2013)}]{Huelga_2013}%
  \BibitemOpen
  \bibfield  {author} {\bibinfo {author} {\bibfnamefont {S.~F.}\ \bibnamefont
  {Huelga}}\ and\ \bibinfo {author} {\bibfnamefont {M.~B.}\ \bibnamefont
  {Plenio}},\ }\bibfield  {title} {\enquote {\bibinfo {title} {Vibrations,
  quanta and biology},}\ }\href
  {https://doi.org/10.1080%2F00405000.2013.829687} {\bibfield  {journal}
  {\bibinfo  {journal} {Contemporary Physics}\ }\textbf {\bibinfo {volume}
  {54}},\ \bibinfo {pages} {181} (\bibinfo {year} {2013})}\BibitemShut
  {NoStop}%
\bibitem [{\citenamefont {Abbott}\ \emph {et~al.}(2008)\citenamefont {Abbott},
  \citenamefont {Davies},\ and\ \citenamefont {Pati}}]{alma9916559633502466}%
  \BibitemOpen
  \bibfield  {author} {\bibinfo {author} {\bibfnamefont {D.}~\bibnamefont
  {Abbott}}, \bibinfo {author} {\bibfnamefont {C.~P.~W.}\ \bibnamefont
  {Davies}}, \ and\ \bibinfo {author} {\bibfnamefont {A.~K.}\ \bibnamefont
  {Pati}},\ }\href {https://doi.org/10.1142/p581} {\emph {\bibinfo {title}
  {Quantum aspects of life}}}\ (\bibinfo  {publisher} {Imperial College
  Press},\ \bibinfo {address} {London},\ \bibinfo {year} {2008})\BibitemShut
  {NoStop}%
\bibitem [{\citenamefont {Rodríguez-Rosario}\ \emph
  {et~al.}(2013)\citenamefont {Rodríguez-Rosario}, \citenamefont
  {Frauenheim},\ and\ \citenamefont {Aspuru-Guzik}}]{Rosario13}%
  \BibitemOpen
  \bibfield  {author} {\bibinfo {author} {\bibfnamefont {C.~A.}\ \bibnamefont
  {Rodríguez-Rosario}}, \bibinfo {author} {\bibfnamefont {T.}~\bibnamefont
  {Frauenheim}}, \ and\ \bibinfo {author} {\bibfnamefont {A.}~\bibnamefont
  {Aspuru-Guzik}},\ }\bibfield  {title} {\enquote {\bibinfo {title}
  {Thermodynamics of quantum coherence},}\ }\href
  {https://doi.org/10.48550/arXiv.1308.1245} {\bibfield  {journal} {\bibinfo
  {journal} {arXiv:1308.1245}\ } (\bibinfo {year} {2013})}\BibitemShut
  {NoStop}%
\bibitem [{\citenamefont {Lostaglio}\ \emph
  {et~al.}(2015{\natexlab{a}})\citenamefont {Lostaglio}, \citenamefont
  {Jennings},\ and\ \citenamefont {Rudolph}}]{Lostaglio_2015}%
  \BibitemOpen
  \bibfield  {author} {\bibinfo {author} {\bibfnamefont {M.}~\bibnamefont
  {Lostaglio}}, \bibinfo {author} {\bibfnamefont {D.}~\bibnamefont {Jennings}},
  \ and\ \bibinfo {author} {\bibfnamefont {T.}~\bibnamefont {Rudolph}},\
  }\bibfield  {title} {\enquote {\bibinfo {title} {Description of quantum
  coherence in thermodynamic processes requires constraints beyond free
  energy},}\ }\href {https://doi.org/10.1038%2Fncomms7383} {\bibfield
  {journal} {\bibinfo  {journal} {Nature Communications}\ }\textbf {\bibinfo
  {volume} {6}},\ \bibinfo {pages} {6383} (\bibinfo {year}
  {2015}{\natexlab{a}})}\BibitemShut {NoStop}%
\bibitem [{\citenamefont {Lostaglio}\ \emph
  {et~al.}(2015{\natexlab{b}})\citenamefont {Lostaglio}, \citenamefont
  {Korzekwa}, \citenamefont {Jennings},\ and\ \citenamefont
  {Rudolph}}]{PhysRevX.5.021001}%
  \BibitemOpen
  \bibfield  {author} {\bibinfo {author} {\bibfnamefont {M.}~\bibnamefont
  {Lostaglio}}, \bibinfo {author} {\bibfnamefont {K.}~\bibnamefont {Korzekwa}},
  \bibinfo {author} {\bibfnamefont {D.}~\bibnamefont {Jennings}}, \ and\
  \bibinfo {author} {\bibfnamefont {T.}~\bibnamefont {Rudolph}},\ }\bibfield
  {title} {\enquote {\bibinfo {title} {Quantum coherence, time-translation
  symmetry, and thermodynamics},}\ }\href
  {https://link.aps.org/doi/10.1103/PhysRevX.5.021001} {\bibfield  {journal}
  {\bibinfo  {journal} {Phys. Rev. X}\ }\textbf {\bibinfo {volume} {5}},\
  \bibinfo {pages} {021001} (\bibinfo {year} {2015}{\natexlab{b}})}\BibitemShut
  {NoStop}%
\bibitem [{\citenamefont {Gardas}\ and\ \citenamefont
  {Deffner}(2015)}]{PhysRevE.92.042126}%
  \BibitemOpen
  \bibfield  {author} {\bibinfo {author} {\bibfnamefont {B.}~\bibnamefont
  {Gardas}}\ and\ \bibinfo {author} {\bibfnamefont {S.}~\bibnamefont
  {Deffner}},\ }\bibfield  {title} {\enquote {\bibinfo {title} {Thermodynamic
  universality of quantum carnot engines},}\ }\href
  {https://link.aps.org/doi/10.1103/PhysRevE.92.042126} {\bibfield  {journal}
  {\bibinfo  {journal} {Phys. Rev. E}\ }\textbf {\bibinfo {volume} {92}},\
  \bibinfo {pages} {042126} (\bibinfo {year} {2015})}\BibitemShut {NoStop}%
\bibitem [{\citenamefont {{\AA}berg}(2006)}]{Aberg06}%
  \BibitemOpen
  \bibfield  {author} {\bibinfo {author} {\bibfnamefont {J.}~\bibnamefont
  {{\AA}berg}},\ }\bibfield  {title} {\enquote {\bibinfo {title} {Quantifying
  superposition},}\ }\href {https://arxiv.org/abs/quant-ph/0612146} {\bibfield
  {journal} {\bibinfo  {journal} {arXiv:quant-ph/0612146}\ } (\bibinfo {year}
  {2006})}\BibitemShut {NoStop}%
\bibitem [{\citenamefont {Baumgratz}\ \emph {et~al.}(2014)\citenamefont
  {Baumgratz}, \citenamefont {Cramer},\ and\ \citenamefont {Plenio}}]{l1-norm}%
  \BibitemOpen
  \bibfield  {author} {\bibinfo {author} {\bibfnamefont {T.}~\bibnamefont
  {Baumgratz}}, \bibinfo {author} {\bibfnamefont {M.}~\bibnamefont {Cramer}}, \
  and\ \bibinfo {author} {\bibfnamefont {M.~B.}\ \bibnamefont {Plenio}},\
  }\bibfield  {title} {\enquote {\bibinfo {title} {Quantifying coherence},}\
  }\href {https://link.aps.org/doi/10.1103/PhysRevLett.113.140401} {\bibfield
  {journal} {\bibinfo  {journal} {Phys. Rev. Lett.}\ }\textbf {\bibinfo
  {volume} {113}},\ \bibinfo {pages} {140401} (\bibinfo {year}
  {2014})}\BibitemShut {NoStop}%
\bibitem [{\citenamefont {Winter}\ and\ \citenamefont {Yang}(2016)}]{Winter16}%
  \BibitemOpen
  \bibfield  {author} {\bibinfo {author} {\bibfnamefont {A.}~\bibnamefont
  {Winter}}\ and\ \bibinfo {author} {\bibfnamefont {D.}~\bibnamefont {Yang}},\
  }\bibfield  {title} {\enquote {\bibinfo {title} {Operational resource theory
  of coherence},}\ }\href {\doibase 10.1103/PhysRevLett.116.120404} {\bibfield
  {journal} {\bibinfo  {journal} {Phys. Rev. Lett.}\ }\textbf {\bibinfo
  {volume} {116}},\ \bibinfo {pages} {120404} (\bibinfo {year}
  {2016})}\BibitemShut {NoStop}%
\bibitem [{\citenamefont {Theurer}\ \emph {et~al.}(2017)\citenamefont
  {Theurer}, \citenamefont {Killoran}, \citenamefont {Egloff},\ and\
  \citenamefont {Plenio}}]{plenio-coherence}%
  \BibitemOpen
  \bibfield  {author} {\bibinfo {author} {\bibfnamefont {T.}~\bibnamefont
  {Theurer}}, \bibinfo {author} {\bibfnamefont {N.}~\bibnamefont {Killoran}},
  \bibinfo {author} {\bibfnamefont {D.}~\bibnamefont {Egloff}}, \ and\ \bibinfo
  {author} {\bibfnamefont {M.~B.}\ \bibnamefont {Plenio}},\ }\bibfield  {title}
  {\enquote {\bibinfo {title} {Resource theory of superposition},}\ }\href
  {https://link.aps.org/doi/10.1103/PhysRevLett.119.230401} {\bibfield
  {journal} {\bibinfo  {journal} {Phys. Rev. Lett.}\ }\textbf {\bibinfo
  {volume} {119}},\ \bibinfo {pages} {230401} (\bibinfo {year}
  {2017})}\BibitemShut {NoStop}%
\bibitem [{\citenamefont {Bischof}\ \emph {et~al.}(2019)\citenamefont
  {Bischof}, \citenamefont {Kampermann},\ and\ \citenamefont
  {Bru\ss{}}}]{Bischof19}%
  \BibitemOpen
  \bibfield  {author} {\bibinfo {author} {\bibfnamefont {F.}~\bibnamefont
  {Bischof}}, \bibinfo {author} {\bibfnamefont {H.}~\bibnamefont {Kampermann}},
  \ and\ \bibinfo {author} {\bibfnamefont {D.}~\bibnamefont {Bru\ss{}}},\
  }\bibfield  {title} {\enquote {\bibinfo {title} {Resource theory of coherence
  based on positive-operator-valued measures},}\ }\href {\doibase
  10.1103/PhysRevLett.123.110402} {\bibfield  {journal} {\bibinfo  {journal}
  {Phys. Rev. Lett.}\ }\textbf {\bibinfo {volume} {123}},\ \bibinfo {pages}
  {110402} (\bibinfo {year} {2019})}\BibitemShut {NoStop}%
\bibitem [{\citenamefont {Srivastava}\ \emph {et~al.}(2021)\citenamefont
  {Srivastava}, \citenamefont {Das},\ and\ \citenamefont {Sen}}]{Srivastava21}%
  \BibitemOpen
  \bibfield  {author} {\bibinfo {author} {\bibfnamefont {C.}~\bibnamefont
  {Srivastava}}, \bibinfo {author} {\bibfnamefont {S.}~\bibnamefont {Das}}, \
  and\ \bibinfo {author} {\bibfnamefont {U.}~\bibnamefont {Sen}},\ }\bibfield
  {title} {\enquote {\bibinfo {title} {Resource theory of quantum coherence
  with probabilistically nondistinguishable pointers and corresponding
  wave-particle duality},}\ }\href {\doibase 10.1103/PhysRevA.103.022417}
  {\bibfield  {journal} {\bibinfo  {journal} {Phys. Rev. A}\ }\textbf {\bibinfo
  {volume} {103}},\ \bibinfo {pages} {022417} (\bibinfo {year}
  {2021})}\BibitemShut {NoStop}%
\bibitem [{\citenamefont {Das}\ \emph {et~al.}(2020)\citenamefont {Das},
  \citenamefont {Mukhopadhyay}, \citenamefont {Roy}, \citenamefont
  {Bhattacharya}, \citenamefont {Sen(De)},\ and\ \citenamefont
  {Sen}}]{sreetama-coherence}%
  \BibitemOpen
  \bibfield  {author} {\bibinfo {author} {\bibfnamefont {S.}~\bibnamefont
  {Das}}, \bibinfo {author} {\bibfnamefont {C.}~\bibnamefont {Mukhopadhyay}},
  \bibinfo {author} {\bibfnamefont {S.~S.}\ \bibnamefont {Roy}}, \bibinfo
  {author} {\bibfnamefont {S.}~\bibnamefont {Bhattacharya}}, \bibinfo {author}
  {\bibfnamefont {A.}~\bibnamefont {Sen(De)}}, \ and\ \bibinfo {author}
  {\bibfnamefont {U.}~\bibnamefont {Sen}},\ }\href
  {https://dx.doi.org/10.1088/1751-8121/ab741f} {\bibfield  {journal} {\bibinfo
   {journal} {Journal of Physics A: Mathematical and Theoretical}\ }\textbf
  {\bibinfo {volume} {53}},\ \bibinfo {pages} {115301} (\bibinfo {year}
  {2020})}\BibitemShut {NoStop}%
\bibitem [{\citenamefont {Banerjee}\ \emph {et~al.}(2021)\citenamefont
  {Banerjee}, \citenamefont {Sen}, \citenamefont {Srivastava},\ and\
  \citenamefont {Sen}}]{Banerjee21}%
  \BibitemOpen
  \bibfield  {author} {\bibinfo {author} {\bibfnamefont {I.}~\bibnamefont
  {Banerjee}}, \bibinfo {author} {\bibfnamefont {K.}~\bibnamefont {Sen}},
  \bibinfo {author} {\bibfnamefont {C.}~\bibnamefont {Srivastava}}, \ and\
  \bibinfo {author} {\bibfnamefont {U.}~\bibnamefont {Sen}},\ }\bibfield
  {title} {\enquote {\bibinfo {title} {Quantum coherence with incomplete set of
  pointers and corresponding wave-particle duality},}\ }\href
  {https://arxiv.org/abs/2108.05849} {\bibfield  {journal} {\bibinfo  {journal}
  {arXiv:2108.05849}\ } (\bibinfo {year} {2021})}\BibitemShut {NoStop}%
\bibitem [{\citenamefont {Wootters}\ and\ \citenamefont
  {Fields}(1989)}]{Wootters89}%
  \BibitemOpen
  \bibfield  {author} {\bibinfo {author} {\bibfnamefont {W.~K.}\ \bibnamefont
  {Wootters}}\ and\ \bibinfo {author} {\bibfnamefont {B.~D}\ \bibnamefont
  {Fields}},\ }\bibfield  {title} {\enquote {\bibinfo {title} {Optimal
  state-determination by mutually unbiased measurements},}\ }\href {\doibase
  https://doi.org/10.1016/0003-4916(89)90322-9} {\bibfield  {journal} {\bibinfo
   {journal} {Annals of Physics}\ }\textbf {\bibinfo {volume} {191}},\ \bibinfo
  {pages} {363} (\bibinfo {year} {1989})}\BibitemShut {NoStop}%
\bibitem [{\citenamefont {Bandyopadhyay}\ \emph {et~al.}(2002)\citenamefont
  {Bandyopadhyay}, \citenamefont {Boykin}, \citenamefont {Roychowdhury},\ and\
  \citenamefont {Vatan}}]{vwani-somshubhro-mub}%
  \BibitemOpen
  \bibfield  {author} {\bibinfo {author} {\bibfnamefont {S.}~\bibnamefont
  {Bandyopadhyay}}, \bibinfo {author} {\bibfnamefont {P.~O.}\ \bibnamefont
  {Boykin}}, \bibinfo {author} {\bibfnamefont {V.}~\bibnamefont
  {Roychowdhury}}, \ and\ \bibinfo {author} {\bibfnamefont {F.}~\bibnamefont
  {Vatan}},\ }\href {https://doi.org/10.1007/s00453-002-0980-7} {\bibfield
  {journal} {\bibinfo  {journal} {Algorithmica}\ }\textbf {\bibinfo {volume}
  {34}},\ \bibinfo {pages} {512} (\bibinfo {year} {2002})}\BibitemShut
  {NoStop}%
\bibitem [{\citenamefont {Planat}\ \emph {et~al.}(2006)\citenamefont {Planat},
  \citenamefont {Rosu},\ and\ \citenamefont {Perrine}}]{planat06}%
  \BibitemOpen
  \bibfield  {author} {\bibinfo {author} {\bibfnamefont {M.}~\bibnamefont
  {Planat}}, \bibinfo {author} {\bibfnamefont {H.~C.}\ \bibnamefont {Rosu}}, \
  and\ \bibinfo {author} {\bibfnamefont {S.}~\bibnamefont {Perrine}},\
  }\bibfield  {title} {\enquote {\bibinfo {title} {A survey of finite algebraic
  geometrical structures underlying mutually unbiased quantum measurements},}\
  }\href {\doibase 10.1007/s10701-006-9079-3} {\bibfield  {journal} {\bibinfo
  {journal} {Foundations of Physics}\ }\textbf {\bibinfo {volume} {36}},\
  \bibinfo {pages} {1662} (\bibinfo {year} {2006})}\BibitemShut {NoStop}%
\bibitem [{\citenamefont {Bengtsson}(2007)}]{MUB-1}%
  \BibitemOpen
  \bibfield  {author} {\bibinfo {author} {\bibfnamefont {I.}~\bibnamefont
  {Bengtsson}},\ }\bibfield  {title} {\enquote {\bibinfo {title} {Three ways to
  look at mutually unbiased bases},}\ }\href
  {https://doi.org/10.1063/1.2713445} {\bibfield  {journal} {\bibinfo
  {journal} {AIP Conference Proceedings}\ }\textbf {\bibinfo {volume} {889}},\
  \bibinfo {pages} {40} (\bibinfo {year} {2007})}\BibitemShut {NoStop}%
\bibitem [{\citenamefont {Cheng}\ and\ \citenamefont {Hall}(2015)}]{hall}%
  \BibitemOpen
  \bibfield  {author} {\bibinfo {author} {\bibfnamefont {S.}~\bibnamefont
  {Cheng}}\ and\ \bibinfo {author} {\bibfnamefont {M.~J.~W.}\ \bibnamefont
  {Hall}},\ }\bibfield  {title} {\enquote {\bibinfo {title} {Complementarity
  relations for quantum coherence},}\ }\href
  {https://link.aps.org/doi/10.1103/PhysRevA.92.042101} {\bibfield  {journal}
  {\bibinfo  {journal} {Phys. Rev. A}\ }\textbf {\bibinfo {volume} {92}},\
  \bibinfo {pages} {042101} (\bibinfo {year} {2015})}\BibitemShut {NoStop}%
\bibitem [{\citenamefont {Mondal}\ \emph {et~al.}(2017)\citenamefont {Mondal},
  \citenamefont {Pramanik},\ and\ \citenamefont {Pati}}]{pati}%
  \BibitemOpen
  \bibfield  {author} {\bibinfo {author} {\bibfnamefont {D.}~\bibnamefont
  {Mondal}}, \bibinfo {author} {\bibfnamefont {T.}~\bibnamefont {Pramanik}}, \
  and\ \bibinfo {author} {\bibfnamefont {A.}~\bibnamefont {Pati}},\ }\bibfield
  {title} {\enquote {\bibinfo {title} {Nonlocal advantage of quantum
  coherence},}\ }\href {https://link.aps.org/doi/10.1103/PhysRevA.95.010301}
  {\bibfield  {journal} {\bibinfo  {journal} {Physical Review A}\ }\textbf
  {\bibinfo {volume} {95}},\ \bibinfo {pages} {010301(R)} (\bibinfo {year}
  {2017})}\BibitemShut {NoStop}%
\bibitem [{\citenamefont {Mondal}\ and\ \citenamefont
  {Kaszlikowski}(2018)}]{debasis}%
  \BibitemOpen
  \bibfield  {author} {\bibinfo {author} {\bibfnamefont {D.}~\bibnamefont
  {Mondal}}\ and\ \bibinfo {author} {\bibfnamefont {D.}~\bibnamefont
  {Kaszlikowski}},\ }\bibfield  {title} {\enquote {\bibinfo {title}
  {Complementarity relations between quantum steering criteria},}\ }\href
  {https://link.aps.org/doi/10.1103/PhysRevA.98.052330} {\bibfield  {journal}
  {\bibinfo  {journal} {Phys. Rev. A}\ }\textbf {\bibinfo {volume} {98}},\
  \bibinfo {pages} {052330} (\bibinfo {year} {2018})}\BibitemShut {NoStop}%
\bibitem [{\citenamefont {Coffman}\ \emph {et~al.}(2000)\citenamefont
  {Coffman}, \citenamefont {Kundu},\ and\ \citenamefont
  {Wootters}}]{ent-monogamy}%
  \BibitemOpen
  \bibfield  {author} {\bibinfo {author} {\bibfnamefont {V.}~\bibnamefont
  {Coffman}}, \bibinfo {author} {\bibfnamefont {J.}~\bibnamefont {Kundu}}, \
  and\ \bibinfo {author} {\bibfnamefont {W.~K.}\ \bibnamefont {Wootters}},\
  }\bibfield  {title} {\enquote {\bibinfo {title} {Distributed entanglement},}\
  }\href {https://link.aps.org/doi/10.1103/PhysRevA.61.052306} {\bibfield
  {journal} {\bibinfo  {journal} {Phys. Rev. A}\ }\textbf {\bibinfo {volume}
  {61}},\ \bibinfo {pages} {052306} (\bibinfo {year} {2000})}\BibitemShut
  {NoStop}%
\bibitem [{\citenamefont {Koashi}\ and\ \citenamefont
  {Winter}(2004)}]{mono-winter}%
  \BibitemOpen
  \bibfield  {author} {\bibinfo {author} {\bibfnamefont {M.}~\bibnamefont
  {Koashi}}\ and\ \bibinfo {author} {\bibfnamefont {A.}~\bibnamefont
  {Winter}},\ }\bibfield  {title} {\enquote {\bibinfo {title} {Monogamy of
  quantum entanglement and other correlations},}\ }\href
  {https://link.aps.org/doi/10.1103/PhysRevA.69.022309} {\bibfield  {journal}
  {\bibinfo  {journal} {Phys. Rev. A}\ }\textbf {\bibinfo {volume} {69}},\
  \bibinfo {pages} {022309} (\bibinfo {year} {2004})}\BibitemShut {NoStop}%
\bibitem [{\citenamefont {Adesso}\ \emph {et~al.}(2006)\citenamefont {Adesso},
  \citenamefont {Serafini},\ and\ \citenamefont {Illuminati}}]{mono-adesso}%
  \BibitemOpen
  \bibfield  {author} {\bibinfo {author} {\bibfnamefont {G.}~\bibnamefont
  {Adesso}}, \bibinfo {author} {\bibfnamefont {A.}~\bibnamefont {Serafini}}, \
  and\ \bibinfo {author} {\bibfnamefont {F.}~\bibnamefont {Illuminati}},\
  }\bibfield  {title} {\enquote {\bibinfo {title} {Multipartite entanglement in
  three-mode {G}aussian states of continuous-variable systems: Quantification,
  sharing structure, and decoherence},}\ }\href
  {https://link.aps.org/doi/10.1103/PhysRevA.73.032345} {\bibfield  {journal}
  {\bibinfo  {journal} {Phys. Rev. A}\ }\textbf {\bibinfo {volume} {73}},\
  \bibinfo {pages} {032345} (\bibinfo {year} {2006})}\BibitemShut {NoStop}%
\bibitem [{\citenamefont {Osborne}\ and\ \citenamefont
  {Verstraete}(2006)}]{mono-osborne}%
  \BibitemOpen
  \bibfield  {author} {\bibinfo {author} {\bibfnamefont {T.~J.}\ \bibnamefont
  {Osborne}}\ and\ \bibinfo {author} {\bibfnamefont {F.}~\bibnamefont
  {Verstraete}},\ }\bibfield  {title} {\enquote {\bibinfo {title} {General
  monogamy inequality for bipartite qubit entanglement},}\ }\href
  {https://link.aps.org/doi/10.1103/PhysRevLett.96.220503} {\bibfield
  {journal} {\bibinfo  {journal} {Phys. Rev. Lett.}\ }\textbf {\bibinfo
  {volume} {96}},\ \bibinfo {pages} {220503} (\bibinfo {year}
  {2006})}\BibitemShut {NoStop}%
\bibitem [{\citenamefont {Hayashi}\ and\ \citenamefont
  {Chen}(2011)}]{mono-chen}%
  \BibitemOpen
  \bibfield  {author} {\bibinfo {author} {\bibfnamefont {M.}~\bibnamefont
  {Hayashi}}\ and\ \bibinfo {author} {\bibfnamefont {L.}~\bibnamefont {Chen}},\
  }\bibfield  {title} {\enquote {\bibinfo {title} {Weaker entanglement between
  two parties guarantees stronger entanglement with a third party},}\ }\href
  {https://link.aps.org/doi/10.1103/PhysRevA.84.012325} {\bibfield  {journal}
  {\bibinfo  {journal} {Phys. Rev. A}\ }\textbf {\bibinfo {volume} {84}},\
  \bibinfo {pages} {012325} (\bibinfo {year} {2011})}\BibitemShut {NoStop}%
\bibitem [{\citenamefont {Lee}\ and\ \citenamefont {Park}(2009)}]{mono-lee}%
  \BibitemOpen
  \bibfield  {author} {\bibinfo {author} {\bibfnamefont {S.}~\bibnamefont
  {Lee}}\ and\ \bibinfo {author} {\bibfnamefont {J.}~\bibnamefont {Park}},\
  }\bibfield  {title} {\enquote {\bibinfo {title} {Monogamy of entanglement and
  teleportation capability},}\ }\href
  {https://link.aps.org/doi/10.1103/PhysRevA.79.054309} {\bibfield  {journal}
  {\bibinfo  {journal} {Phys. Rev. A}\ }\textbf {\bibinfo {volume} {79}},\
  \bibinfo {pages} {054309} (\bibinfo {year} {2009})}\BibitemShut {NoStop}%
\bibitem [{\citenamefont {Prabhu}\ \emph {et~al.}(2012)\citenamefont {Prabhu},
  \citenamefont {Pati}, \citenamefont {Sen~(De)},\ and\ \citenamefont
  {Sen}}]{discord-monogamy}%
  \BibitemOpen
  \bibfield  {author} {\bibinfo {author} {\bibfnamefont {R.}~\bibnamefont
  {Prabhu}}, \bibinfo {author} {\bibfnamefont {A.~K.}\ \bibnamefont {Pati}},
  \bibinfo {author} {\bibfnamefont {A.}~\bibnamefont {Sen~(De)}}, \ and\
  \bibinfo {author} {\bibfnamefont {U.}~\bibnamefont {Sen}},\ }\bibfield
  {title} {\enquote {\bibinfo {title} {Conditions for monogamy of quantum
  correlations: {G}reenberger-{H}orne-{Z}eilinger versus {W} states},}\ }\href
  {https://link.aps.org/doi/10.1103/PhysRevA.85.040102} {\bibfield  {journal}
  {\bibinfo  {journal} {Phys. Rev. A}\ }\textbf {\bibinfo {volume} {85}},\
  \bibinfo {pages} {040102} (\bibinfo {year} {2012})}\BibitemShut {NoStop}%
\bibitem [{\citenamefont {Guo}\ \emph {et~al.}(2021)\citenamefont {Guo},
  \citenamefont {Huang},\ and\ \citenamefont {Zhang}}]{mono-discord-2}%
  \BibitemOpen
  \bibfield  {author} {\bibinfo {author} {\bibfnamefont {Y.}~\bibnamefont
  {Guo}}, \bibinfo {author} {\bibfnamefont {L.}~\bibnamefont {Huang}}, \ and\
  \bibinfo {author} {\bibfnamefont {Y.}~\bibnamefont {Zhang}},\ }\bibfield
  {title} {\enquote {\bibinfo {title} {Monogamy of quantum discord},}\ }\href
  {https://doi.org/10.1088/2058-9565/ac26b0} {\bibfield  {journal} {\bibinfo
  {journal} {Quantum Science and Technology}\ }\textbf {\bibinfo {volume}
  {6}},\ \bibinfo {pages} {045028} (\bibinfo {year} {2021})}\BibitemShut
  {NoStop}%
\bibitem [{\citenamefont {Horodecki}\ \emph {et~al.}(2009)\citenamefont
  {Horodecki}, \citenamefont {Horodecki}, \citenamefont {Horodecki},\ and\
  \citenamefont {Horodecki}}]{ent-review-1}%
  \BibitemOpen
  \bibfield  {author} {\bibinfo {author} {\bibfnamefont {R.}~\bibnamefont
  {Horodecki}}, \bibinfo {author} {\bibfnamefont {P.}~\bibnamefont
  {Horodecki}}, \bibinfo {author} {\bibfnamefont {M.}~\bibnamefont
  {Horodecki}}, \ and\ \bibinfo {author} {\bibfnamefont {K.}~\bibnamefont
  {Horodecki}},\ }\bibfield  {title} {\enquote {\bibinfo {title} {Quantum
  entanglement},}\ }\href {https://link.aps.org/doi/10.1103/RevModPhys.81.865}
  {\bibfield  {journal} {\bibinfo  {journal} {Rev. Mod. Phys.}\ }\textbf
  {\bibinfo {volume} {81}},\ \bibinfo {pages} {865} (\bibinfo {year}
  {2009})}\BibitemShut {NoStop}%
\bibitem [{\citenamefont {Modi}\ \emph {et~al.}(2012)\citenamefont {Modi},
  \citenamefont {Brodutch}, \citenamefont {Cable}, \citenamefont {Paterek},\
  and\ \citenamefont {Vedral}}]{kavan-discord-review}%
  \BibitemOpen
  \bibfield  {author} {\bibinfo {author} {\bibfnamefont {K.}~\bibnamefont
  {Modi}}, \bibinfo {author} {\bibfnamefont {A.}~\bibnamefont {Brodutch}},
  \bibinfo {author} {\bibfnamefont {H.}~\bibnamefont {Cable}}, \bibinfo
  {author} {\bibfnamefont {T.}~\bibnamefont {Paterek}}, \ and\ \bibinfo
  {author} {\bibfnamefont {V.}~\bibnamefont {Vedral}},\ }\bibfield  {title}
  {\enquote {\bibinfo {title} {The classical-quantum boundary for correlations:
  Discord and related measures},}\ }\href
  {https://link.aps.org/doi/10.1103/RevModPhys.84.1655} {\bibfield  {journal}
  {\bibinfo  {journal} {Rev. Mod. Phys.}\ }\textbf {\bibinfo {volume} {84}},\
  \bibinfo {pages} {1655} (\bibinfo {year} {2012})}\BibitemShut {NoStop}%
\bibitem [{\citenamefont {Bera}\ \emph {et~al.}(2017)\citenamefont {Bera},
  \citenamefont {Das}, \citenamefont {Sadhukhan}, \citenamefont {Roy},
  \citenamefont {Sen(De)},\ and\ \citenamefont {Sen}}]{discord-review-1}%
  \BibitemOpen
  \bibfield  {author} {\bibinfo {author} {\bibfnamefont {A.}~\bibnamefont
  {Bera}}, \bibinfo {author} {\bibfnamefont {T.}~\bibnamefont {Das}}, \bibinfo
  {author} {\bibfnamefont {D.}~\bibnamefont {Sadhukhan}}, \bibinfo {author}
  {\bibfnamefont {S.~S.}\ \bibnamefont {Roy}}, \bibinfo {author} {\bibfnamefont
  {A.}~\bibnamefont {Sen(De)}}, \ and\ \bibinfo {author} {\bibfnamefont
  {U.}~\bibnamefont {Sen}},\ }\bibfield  {title} {\enquote {\bibinfo {title}
  {Quantum discord and its allies: a review of recent progress},}\ }\href
  {https://dx.doi.org/10.1088/1361-6633/aa872f} {\bibfield  {journal} {\bibinfo
   {journal} {Reports on Progress in Physics}\ }\textbf {\bibinfo {volume}
  {81}},\ \bibinfo {pages} {024001} (\bibinfo {year} {2017})}\BibitemShut
  {NoStop}%
\bibitem [{\citenamefont {Kim}\ \emph {et~al.}(2012)\citenamefont {Kim},
  \citenamefont {Gour},\ and\ \citenamefont {Sanders}}]{sanders-ent-review}%
  \BibitemOpen
  \bibfield  {author} {\bibinfo {author} {\bibfnamefont {J.~S.}\ \bibnamefont
  {Kim}}, \bibinfo {author} {\bibfnamefont {G.}~\bibnamefont {Gour}}, \ and\
  \bibinfo {author} {\bibfnamefont {B.~C.}\ \bibnamefont {Sanders}},\
  }\bibfield  {title} {\enquote {\bibinfo {title} {Limitations to sharing
  entanglement},}\ }\href {https://doi.org/10.1080/00107514.2012.725560}
  {\bibfield  {journal} {\bibinfo  {journal} {Contemporary Physics}\ }\textbf
  {\bibinfo {volume} {53}},\ \bibinfo {pages} {417} (\bibinfo {year}
  {2012})}\BibitemShut {NoStop}%
\bibitem [{\citenamefont {Dhar}\ \emph {et~al.}(2017)\citenamefont {Dhar},
  \citenamefont {Pal}, \citenamefont {Rakshit}, \citenamefont {Sen(De)},\ and\
  \citenamefont {Sen}}]{discord-monogamy-review-ujjwal}%
  \BibitemOpen
  \bibfield  {author} {\bibinfo {author} {\bibfnamefont {H.~S.}\ \bibnamefont
  {Dhar}}, \bibinfo {author} {\bibfnamefont {A.~K.}\ \bibnamefont {Pal}},
  \bibinfo {author} {\bibfnamefont {D.}~\bibnamefont {Rakshit}}, \bibinfo
  {author} {\bibfnamefont {A.}~\bibnamefont {Sen(De)}}, \ and\ \bibinfo
  {author} {\bibfnamefont {U.}~\bibnamefont {Sen}},\ }\enquote {\bibinfo
  {title} {Monogamy of quantum correlations-a review},}\ in\ \href
  {https://doi.org/10.1007/978-3-319-53412-1} {\emph {\bibinfo {booktitle}
  {Lectures on General Quantum Correlations and their Applications}}},\
  \bibinfo {editor} {edited by\ \bibinfo {editor} {\bibfnamefont {F.~F.}\
  \bibnamefont {Fanchini}}, \bibinfo {editor} {\bibfnamefont {D.}~\bibnamefont
  {Pinto}}, \ and\ \bibinfo {editor} {\bibfnamefont {G.}~\bibnamefont
  {Adesso}}}\ (\bibinfo  {publisher} {Springer International Publishing},\
  \bibinfo {year} {2017})\ pp.\ \bibinfo {pages} {23--64}\BibitemShut {NoStop}%
\bibitem [{\citenamefont {Streltsov}\ \emph {et~al.}(2015)\citenamefont
  {Streltsov}, \citenamefont {Singh}, \citenamefont {Dhar}, \citenamefont
  {Bera},\ and\ \citenamefont {Adesso}}]{coherence-streltsov}%
  \BibitemOpen
  \bibfield  {author} {\bibinfo {author} {\bibfnamefont {A.}~\bibnamefont
  {Streltsov}}, \bibinfo {author} {\bibfnamefont {U.}~\bibnamefont {Singh}},
  \bibinfo {author} {\bibfnamefont {H.~S.}\ \bibnamefont {Dhar}}, \bibinfo
  {author} {\bibfnamefont {M.~N.}\ \bibnamefont {Bera}}, \ and\ \bibinfo
  {author} {\bibfnamefont {G.}~\bibnamefont {Adesso}},\ }\bibfield  {title}
  {\enquote {\bibinfo {title} {Measuring quantum coherence with
  entanglement},}\ }\href
  {https://link.aps.org/doi/10.1103/PhysRevLett.115.020403} {\bibfield
  {journal} {\bibinfo  {journal} {Phys. Rev. Lett.}\ }\textbf {\bibinfo
  {volume} {115}},\ \bibinfo {pages} {020403} (\bibinfo {year}
  {2015})}\BibitemShut {NoStop}%
\bibitem [{\citenamefont {Hu}\ \emph {et~al.}(2018)\citenamefont {Hu},
  \citenamefont {Wang},\ and\ \citenamefont {Fan}}]{bell-naqc}%
  \BibitemOpen
  \bibfield  {author} {\bibinfo {author} {\bibfnamefont {M-L.}\ \bibnamefont
  {Hu}}, \bibinfo {author} {\bibfnamefont {X-M.}\ \bibnamefont {Wang}}, \ and\
  \bibinfo {author} {\bibfnamefont {H.}~\bibnamefont {Fan}},\ }\bibfield
  {title} {\enquote {\bibinfo {title} {Hierarchy of the nonlocal advantage of
  quantum coherence and {B}ell nonlocality},}\ }\href
  {https://link.aps.org/doi/10.1103/PhysRevA.98.032317} {\bibfield  {journal}
  {\bibinfo  {journal} {Phys. Rev. A}\ }\textbf {\bibinfo {volume} {98}},\
  \bibinfo {pages} {032317} (\bibinfo {year} {2018})}\BibitemShut {NoStop}%
\bibitem [{\citenamefont {Ac\'{\i}n}\ \emph {et~al.}(2000)\citenamefont
  {Ac\'{\i}n}, \citenamefont {Andrianov}, \citenamefont {Costa}, \citenamefont
  {Jan\'e}, \citenamefont {Latorre},\ and\ \citenamefont
  {Tarrach}}]{three-pure-1}%
  \BibitemOpen
  \bibfield  {author} {\bibinfo {author} {\bibfnamefont {A.}~\bibnamefont
  {Ac\'{\i}n}}, \bibinfo {author} {\bibfnamefont {A.}~\bibnamefont
  {Andrianov}}, \bibinfo {author} {\bibfnamefont {L.}~\bibnamefont {Costa}},
  \bibinfo {author} {\bibfnamefont {E.}~\bibnamefont {Jan\'e}}, \bibinfo
  {author} {\bibfnamefont {J.~I.}\ \bibnamefont {Latorre}}, \ and\ \bibinfo
  {author} {\bibfnamefont {R.}~\bibnamefont {Tarrach}},\ }\bibfield  {title}
  {\enquote {\bibinfo {title} {Generalized {S}chmidt decomposition and
  classification of three-quantum-bit states},}\ }\href
  {https://link.aps.org/doi/10.1103/PhysRevLett.85.1560} {\bibfield  {journal}
  {\bibinfo  {journal} {Phys. Rev. Lett.}\ }\textbf {\bibinfo {volume} {85}},\
  \bibinfo {pages} {1560} (\bibinfo {year} {2000})}\BibitemShut {NoStop}%
\bibitem [{\citenamefont {D\"ur}\ \emph {et~al.}(2000)\citenamefont {D\"ur},
  \citenamefont {Vidal},\ and\ \citenamefont {Cirac}}]{three-pure-2}%
  \BibitemOpen
  \bibfield  {author} {\bibinfo {author} {\bibfnamefont {W.}~\bibnamefont
  {D\"ur}}, \bibinfo {author} {\bibfnamefont {G.}~\bibnamefont {Vidal}}, \ and\
  \bibinfo {author} {\bibfnamefont {J.~I.}\ \bibnamefont {Cirac}},\ }\bibfield
  {title} {\enquote {\bibinfo {title} {Three qubits can be entangled in two
  inequivalent ways},}\ }\href
  {https://link.aps.org/doi/10.1103/PhysRevA.62.062314} {\bibfield  {journal}
  {\bibinfo  {journal} {Phys. Rev. A}\ }\textbf {\bibinfo {volume} {62}},\
  \bibinfo {pages} {062314} (\bibinfo {year} {2000})}\BibitemShut {NoStop}%
\bibitem [{\citenamefont {Ac\'{\i}n}\ \emph {et~al.}(2001)\citenamefont
  {Ac\'{\i}n}, \citenamefont {Bru\ss{}}, \citenamefont {Lewenstein},\ and\
  \citenamefont {Sanpera}}]{three-pure-3}%
  \BibitemOpen
  \bibfield  {author} {\bibinfo {author} {\bibfnamefont {A.}~\bibnamefont
  {Ac\'{\i}n}}, \bibinfo {author} {\bibfnamefont {D.}~\bibnamefont {Bru\ss{}}},
  \bibinfo {author} {\bibfnamefont {M.}~\bibnamefont {Lewenstein}}, \ and\
  \bibinfo {author} {\bibfnamefont {A.}~\bibnamefont {Sanpera}},\ }\bibfield
  {title} {\enquote {\bibinfo {title} {Classification of mixed three-qubit
  states},}\ }\href {https://link.aps.org/doi/10.1103/PhysRevLett.87.040401}
  {\bibfield  {journal} {\bibinfo  {journal} {Phys. Rev. Lett.}\ }\textbf
  {\bibinfo {volume} {87}},\ \bibinfo {pages} {040401} (\bibinfo {year}
  {2001})}\BibitemShut {NoStop}%
\bibitem [{\citenamefont {Peres}(1996)}]{peres-PPT}%
  \BibitemOpen
  \bibfield  {author} {\bibinfo {author} {\bibfnamefont {A.}~\bibnamefont
  {Peres}},\ }\bibfield  {title} {\enquote {\bibinfo {title} {Separability
  criterion for density matrices},}\ }\href
  {https://link.aps.org/doi/10.1103/PhysRevLett.77.1413} {\bibfield  {journal}
  {\bibinfo  {journal} {Phys. Rev. Lett.}\ }\textbf {\bibinfo {volume} {77}},\
  \bibinfo {pages} {1413} (\bibinfo {year} {1996})}\BibitemShut {NoStop}%
\bibitem [{\citenamefont {Horodecki}\ \emph {et~al.}(1996)\citenamefont
  {Horodecki}, \citenamefont {P.},\ and\ \citenamefont
  {Horodecki}}]{horodecki-PPT}%
  \BibitemOpen
  \bibfield  {author} {\bibinfo {author} {\bibfnamefont {M.}~\bibnamefont
  {Horodecki}}, \bibinfo {author} {\bibfnamefont {Horodecki}\ \bibnamefont
  {P.}}, \ and\ \bibinfo {author} {\bibfnamefont {R.}~\bibnamefont
  {Horodecki}},\ }\bibfield  {title} {\enquote {\bibinfo {title} {Separability
  of mixed states: necessary and sufficient conditions},}\ }\href
  {https://www.sciencedirect.com/science/article/pii/S0375960196007062}
  {\bibfield  {journal} {\bibinfo  {journal} {Physics Letters A}\ }\textbf
  {\bibinfo {volume} {223}},\ \bibinfo {pages} {1} (\bibinfo {year}
  {1996})}\BibitemShut {NoStop}%
\bibitem [{\citenamefont {Wiseman}\ \emph {et~al.}(2007)\citenamefont
  {Wiseman}, \citenamefont {Jones},\ and\ \citenamefont
  {Doherty}}]{werner-different-range}%
  \BibitemOpen
  \bibfield  {author} {\bibinfo {author} {\bibfnamefont {H.~M.}\ \bibnamefont
  {Wiseman}}, \bibinfo {author} {\bibfnamefont {S.~J.}\ \bibnamefont {Jones}},
  \ and\ \bibinfo {author} {\bibfnamefont {A.~C.}\ \bibnamefont {Doherty}},\
  }\bibfield  {title} {\enquote {\bibinfo {title} {Steering, entanglement,
  nonlocality, and the {E}instein-{P}odolsky-{R}osen paradox},}\ }\href
  {https://link.aps.org/doi/10.1103/PhysRevLett.98.140402} {\bibfield
  {journal} {\bibinfo  {journal} {Phys. Rev. Lett.}\ }\textbf {\bibinfo
  {volume} {98}},\ \bibinfo {pages} {140402} (\bibinfo {year}
  {2007})}\BibitemShut {NoStop}%
\bibitem [{\citenamefont {Toner}\ and\ \citenamefont
  {Verstraete}(2006)}]{toner-bell-2}%
  \BibitemOpen
  \bibfield  {author} {\bibinfo {author} {\bibfnamefont {B.}~\bibnamefont
  {Toner}}\ and\ \bibinfo {author} {\bibfnamefont {F.}~\bibnamefont
  {Verstraete}},\ }\bibfield  {title} {\enquote {\bibinfo {title} {Monogamy of
  {B}ell correlations and {T}sirelson's bound},}\ }\href
  {https://doi.org/10.48550/arXiv.quant-ph/0611001} {\bibfield  {journal}
  {\bibinfo  {journal} {arXiv:quant-ph/0611001}\ } (\bibinfo {year}
  {2006})}\BibitemShut {NoStop}%
\bibitem [{\citenamefont {Toner}(2009)}]{toner-bell-1}%
  \BibitemOpen
  \bibfield  {author} {\bibinfo {author} {\bibfnamefont {B.}~\bibnamefont
  {Toner}},\ }\bibfield  {title} {\enquote {\bibinfo {title} {Monogamy of
  non-local quantum correlations},}\ }\href
  {https://doi.org/10.1098/rspa.2008.0149} {\bibfield  {journal} {\bibinfo
  {journal} {Proc. R. Soc. A.}\ }\textbf {\bibinfo {volume} {465}},\ \bibinfo
  {pages} {59} (\bibinfo {year} {2009})}\BibitemShut {NoStop}%
\bibitem [{\citenamefont {Kurzy\ifmmode~\acute{n}\else \'{n}\fi{}ski}\ \emph
  {et~al.}(2011)\citenamefont {Kurzy\ifmmode~\acute{n}\else \'{n}\fi{}ski},
  \citenamefont {Paterek}, \citenamefont {Ramanathan}, \citenamefont
  {Laskowski},\ and\ \citenamefont
  {Kaszlikowski}}]{kazlikowski-bell-monogamy-1}%
  \BibitemOpen
  \bibfield  {author} {\bibinfo {author} {\bibfnamefont {P.}~\bibnamefont
  {Kurzy\ifmmode~\acute{n}\else \'{n}\fi{}ski}}, \bibinfo {author}
  {\bibfnamefont {T.}~\bibnamefont {Paterek}}, \bibinfo {author} {\bibfnamefont
  {R.}~\bibnamefont {Ramanathan}}, \bibinfo {author} {\bibfnamefont
  {W.}~\bibnamefont {Laskowski}}, \ and\ \bibinfo {author} {\bibfnamefont
  {D.}~\bibnamefont {Kaszlikowski}},\ }\bibfield  {title} {\enquote {\bibinfo
  {title} {Correlation complementarity yields bell monogamy relations},}\
  }\href {https://link.aps.org/doi/10.1103/PhysRevLett.106.180402} {\bibfield
  {journal} {\bibinfo  {journal} {Phys. Rev. Lett.}\ }\textbf {\bibinfo
  {volume} {106}},\ \bibinfo {pages} {180402} (\bibinfo {year}
  {2011})}\BibitemShut {NoStop}%
\bibitem [{\citenamefont {Tran}\ \emph {et~al.}(2018)\citenamefont {Tran},
  \citenamefont {Ramanathan}, \citenamefont {McKague}, \citenamefont
  {Kaszlikowski},\ and\ \citenamefont {Paterek}}]{kazlikowski-bell-monogamy-2}%
  \BibitemOpen
  \bibfield  {author} {\bibinfo {author} {\bibfnamefont {M.~C.}\ \bibnamefont
  {Tran}}, \bibinfo {author} {\bibfnamefont {R.}~\bibnamefont {Ramanathan}},
  \bibinfo {author} {\bibfnamefont {M.}~\bibnamefont {McKague}}, \bibinfo
  {author} {\bibfnamefont {D.}~\bibnamefont {Kaszlikowski}}, \ and\ \bibinfo
  {author} {\bibfnamefont {T.}~\bibnamefont {Paterek}},\ }\bibfield  {title}
  {\enquote {\bibinfo {title} {Bell monogamy relations in arbitrary qubit
  networks},}\ }\href {https://link.aps.org/doi/10.1103/PhysRevA.98.052325}
  {\bibfield  {journal} {\bibinfo  {journal} {Phys. Rev. A}\ }\textbf {\bibinfo
  {volume} {98}},\ \bibinfo {pages} {052325} (\bibinfo {year}
  {2018})}\BibitemShut {NoStop}%
\bibitem [{\citenamefont {Prabhu}\ \emph {et~al.}(2013)\citenamefont {Prabhu},
  \citenamefont {Pati}, \citenamefont {Sen(De)},\ and\ \citenamefont
  {Sen}}]{prabhu-dense-coding}%
  \BibitemOpen
  \bibfield  {author} {\bibinfo {author} {\bibfnamefont {R.}~\bibnamefont
  {Prabhu}}, \bibinfo {author} {\bibfnamefont {A.~K.}\ \bibnamefont {Pati}},
  \bibinfo {author} {\bibfnamefont {A.}~\bibnamefont {Sen(De)}}, \ and\
  \bibinfo {author} {\bibfnamefont {U.}~\bibnamefont {Sen}},\ }\bibfield
  {title} {\enquote {\bibinfo {title} {Exclusion principle for quantum dense
  coding},}\ }\href {https://link.aps.org/doi/10.1103/PhysRevA.87.052319}
  {\bibfield  {journal} {\bibinfo  {journal} {Phys. Rev. A}\ }\textbf {\bibinfo
  {volume} {87}},\ \bibinfo {pages} {052319} (\bibinfo {year}
  {2013})}\BibitemShut {NoStop}%
\bibitem [{\citenamefont {Greenberger}\ \emph {et~al.}(1989)\citenamefont
  {Greenberger}, \citenamefont {Horne},\ and\ \citenamefont
  {Zeilinger}}]{ghz-class-1}%
  \BibitemOpen
  \bibfield  {author} {\bibinfo {author} {\bibfnamefont {D.~M.}\ \bibnamefont
  {Greenberger}}, \bibinfo {author} {\bibfnamefont {M.~A.}\ \bibnamefont
  {Horne}}, \ and\ \bibinfo {author} {\bibfnamefont {A.}~\bibnamefont
  {Zeilinger}},\ }\enquote {\bibinfo {title} {Going beyond bell's theorem},}\
  in\ \href {https://doi.org/10.1007/978-94-017-0849-4_10} {\emph {\bibinfo
  {booktitle} {Bell's Theorem, Quantum Theory and Conceptions of the
  Universe}}},\ \bibinfo {editor} {edited by\ \bibinfo {editor} {\bibfnamefont
  {M.}~\bibnamefont {Kafatos}}}\ (\bibinfo  {publisher} {Springer
  Netherlands},\ \bibinfo {address} {Dordrecht},\ \bibinfo {year} {1989})\ pp.\
  \bibinfo {pages} {69--72}\BibitemShut {NoStop}%
\bibitem [{\citenamefont {Mermin}(1990)}]{ghz-class-2}%
  \BibitemOpen
  \bibfield  {author} {\bibinfo {author} {\bibfnamefont {N.~D.}\ \bibnamefont
  {Mermin}},\ }\bibfield  {title} {\enquote {\bibinfo {title} {Extreme quantum
  entanglement in a superposition of macroscopically distinct states},}\ }\href
  {https://link.aps.org/doi/10.1103/PhysRevLett.65.1838} {\bibfield  {journal}
  {\bibinfo  {journal} {Phys. Rev. Lett.}\ }\textbf {\bibinfo {volume} {65}},\
  \bibinfo {pages} {1838} (\bibinfo {year} {1990})}\BibitemShut {NoStop}%
\bibitem [{\citenamefont {Zeilinger}\ \emph {et~al.}(1992)\citenamefont
  {Zeilinger}, \citenamefont {Horne},\ and\ \citenamefont
  {Greenberger}}]{w-class-1}%
  \BibitemOpen
  \bibfield  {author} {\bibinfo {author} {\bibfnamefont {A.}~\bibnamefont
  {Zeilinger}}, \bibinfo {author} {\bibfnamefont {M.~A.}\ \bibnamefont
  {Horne}}, \ and\ \bibinfo {author} {\bibfnamefont {D.~M.}\ \bibnamefont
  {Greenberger}},\ }\bibfield  {title} {\enquote {\bibinfo {title}
  {Higher-order quantum entanglement},}\ }in\ \href@noop {} {\emph {\bibinfo
  {booktitle} {Squeezed States and Uncertainty Relations}}}\ (\bibinfo {year}
  {1992})\ pp.\ \bibinfo {pages} {73--81}\BibitemShut {NoStop}%
\bibitem [{\citenamefont {Uola}\ \emph {et~al.}(2020)\citenamefont {Uola},
  \citenamefont {Costa}, \citenamefont {Nguyen},\ and\ \citenamefont
  {G\"uhne}}]{steer-rev}%
  \BibitemOpen
  \bibfield  {author} {\bibinfo {author} {\bibfnamefont {R.}~\bibnamefont
  {Uola}}, \bibinfo {author} {\bibfnamefont {A.~C.~S.}\ \bibnamefont {Costa}},
  \bibinfo {author} {\bibfnamefont {H.~C.}\ \bibnamefont {Nguyen}}, \ and\
  \bibinfo {author} {\bibfnamefont {O.}~\bibnamefont {G\"uhne}},\ }\bibfield
  {title} {\enquote {\bibinfo {title} {Quantum steering},}\ }\href
  {https://link.aps.org/doi/10.1103/RevModPhys.92.015001} {\bibfield  {journal}
  {\bibinfo  {journal} {Rev. Mod. Phys.}\ }\textbf {\bibinfo {volume} {92}},\
  \bibinfo {pages} {015001} (\bibinfo {year} {2020})}\BibitemShut {NoStop}%
\end{thebibliography}%

\end{document}